\documentclass[lettersize,journal]{IEEEtran}

\IEEEoverridecommandlockouts

\usepackage{amsmath,graphicx,epstopdf,amssymb,amsthm}
\usepackage[most]{tcolorbox}

\usepackage{dsfont}
\usepackage{url}
\usepackage{color}
\usepackage{epstopdf}
\usepackage{verbatim}
\usepackage{cuted}
\usepackage{lipsum,graphicx,subcaption}
\usepackage[labelformat=simple]{subcaption}

\usepackage[utf8]{inputenc}
\usepackage[english]{babel}
\usepackage{caption}
\usepackage{enumitem}
\usepackage[linesnumbered,ruled,vlined]{algorithm2e}
\DeclareCaptionLabelFormat{lc}{\MakeLowercase{#1}~#2}
\usepackage{balance}
\usepackage{float}
\usepackage{mathtools}

\newtheorem{lemma}{Lemma}
\newtheorem{definition}{Definition}

\DeclareMathOperator*{\argmin}{arg\,min}
\allowdisplaybreaks

\usepackage[font=small]{caption}
\usepackage{enumitem}

\usepackage{tabularx,booktabs}
\newcolumntype{Y}{>{\centering\arraybackslash}X}
\usepackage{multirow}
\allowdisplaybreaks

\usepackage{amsmath,bm,thmtools}

\begin{document}

\title{Multi-Source to Multi-Target Decentralized \\ Federated Domain Adaptation}

\author{Su Wang, Seyyedali Hosseinalipour,~\IEEEmembership{Member,~IEEE,} Christopher~G.~Brinton,~\IEEEmembership{Senior~Member,~IEEE}
\thanks{S. Wang and C. G. Brinton are with Purdue University, IN, USA e-mail: \{wang2506,cgb\}@purdue.edu.}
\thanks{S. Hosseinalipour is with University at Buffalo--SUNY, NY, USA email: alipour@buffalo.edu.}
\thanks{This project was supported in part by the Office of Naval Research (ONR) under grants N000142212305 and N00014-23-C-1016, and by the National Science Foundation (NSF) under grants CNS-2146171 and CNS-2212565.}}
\maketitle

\begin{abstract}
    Heterogeneity across devices in federated learning (FL) typically refers to statistical (e.g., non-i.i.d. data distributions) and resource (e.g., communication bandwidth) dimensions. 
    In this paper, we focus on another important dimension that has received less attention: varying quantities/distributions of \textit{labeled} and \textit{unlabeled} data across devices. 
    In order to leverage all data, we develop a decentralized federated domain adaptation methodology which considers the transfer of ML models from devices with high quality labeled data (called sources) to devices with low quality or unlabeled data (called targets). 
    Our methodology, Source-Target Determination and Link Formation (ST-LF), optimizes both (i) classification of devices into sources and targets and (ii) source-target link formation, in a manner that considers the trade-off between ML model accuracy and communication energy efficiency. 
    To obtain a concrete objective function, we derive a measurable generalization error bound that accounts for estimates of source-target hypothesis deviations and divergences between data distributions. 
    The resulting optimization problem is a mixed-integer signomial program, a class of NP-hard problems, for which we develop an algorithm based on successive convex approximations to solve it tractably. 
    Subsequent numerical evaluations of ST-LF demonstrate that it improves classification accuracy and energy efficiency over state-of-the-art baselines.
\end{abstract}


\section{Introduction}
\label{sec:intro}

\noindent Federated learning (FL)~\cite{konevcny2016federated,lim2020federated} is a class of distributed machine learning (ML) methods designed for training ML models across non-independent and identically distributed (non-i.i.d.) datasets at edge devices. 
FL has been actively studied in two different scenarios: (i) centralized FL~\cite{mcmahan2017communication,luo2021cost}, which relies on a server to aggregate ML parameters across devices, and (ii) decentralized FL~\cite{ye2022decentralized,roy2019braintorrent,wang2021device}, in which devices exchange ML parameters through peer-to-peer communications. 
Existing work in both scenarios has made great strides in addressing the fact that devices often exhibit heterogeneity with respect to their local data distributions and/or local communication/computation capabilities for processing model updates.
However, many of these works have implicitly assumed that each device has a sufficient amount of labeled data available to train local models in the first place~\cite{peng2019federated,schmarje2021survey}. 
This tends to be unrealistic, particularly when the measurements are originating from unfamiliar environments. 

{
In training supervised ML models through FL, the resulting global model will naturally be biased to favor devices/distributions with more labeled data. For example, the standard FL algorithm~\cite{mcmahan2017communication} and its weighted averaging rule will yield a global model that favors data categories that are labeled and more populous. 
Since labeled and unlabeled data are both distributed across real-world networks in a non-i.i.d. manner, the global model may not capture the unique statistical properties at devices with mostly unlabeled data, and thus lead to poor performance for those devices.

Consequently, there is a need for FL methodologies that address heterogeneity in the local quantity and distribution of unlabeled data. To address this, consider the possibility of using both labeled and unlabeled data at network devices to calculate statistical divergence measures between devices. 
In decentralized FL, we could then optimize the transfer of unique weighted combinations of ML models from (a) devices with high quality labeled data, called sources or source domains, to (b) devices with poorly labeled and/or unlabeled data, called targets or target domains. 
Such source-to-target model transmissions are commonly called domain adaptation~\cite{ben2010theory,mansour2008domain}. However, literature in this area~\cite{peng2019federated,zhao2018adversarial,zhao2020multi} typically assumes that the sources and targets are known apriori. 
Since devices are heterogeneous with respect to both local \textit{quantity} and \textit{distribution} of unlabeled data in federated settings, optimal source-target determination is non-trivial, e.g., it could yield targets with some labeled data and sources with some unlabeled data.
Moreover, empirical data divergence computations should not require the transmission of raw data. 

In this paper, we investigate domain adaptation in decentralized FL, with the above considerations and additional network factors such as energy efficiency. 
Our proposed methodology, \textit{Source-Target Determination and Link Formation} (ST-LF), is among the first to jointly (i) determine device source/target classifications, (ii) analyze the link formation problem from sources to targets, and (iii) minimize the total network communication resource consumption. 
Individually, (i), (ii), and (iii) pose several challenges: for example, optimal source/target classification may entail considering targets with labeled data.  
The difficulty is further exacerbated as (i), (ii), and (iii) are intertwined together, i.e., the result of source/target selection directly influences possible link formation regimes and communication resource consumption.

ST-LF aims to classify devices as either sources or targets by estimating post-training ML error based on the local quantity and quality of data. 
To facilitate this, we develop theoretical bounds for expected post-training ML error and generalization bounds to capture the contribution of source-to-target model transmissions. 
{\color{black}We also develop an algorithm to compute empirical divergence on labeled and unlabeled data, which offers similar data privacy advantages found in FL by eliminating the need to transmit devices' raw data over the network.} 
Then, we leverage theoretical bounds and divergence measures to optimize link formation from multiple sources to multiple targets.
ST-LF for the first time reveals how communication efficiency requirements influence link formation from sources to targets, and optimizes network communication resource consumption jointly with link formation.}

\begin{figure}[t]
    \centering
    \includegraphics[width=0.48\textwidth]{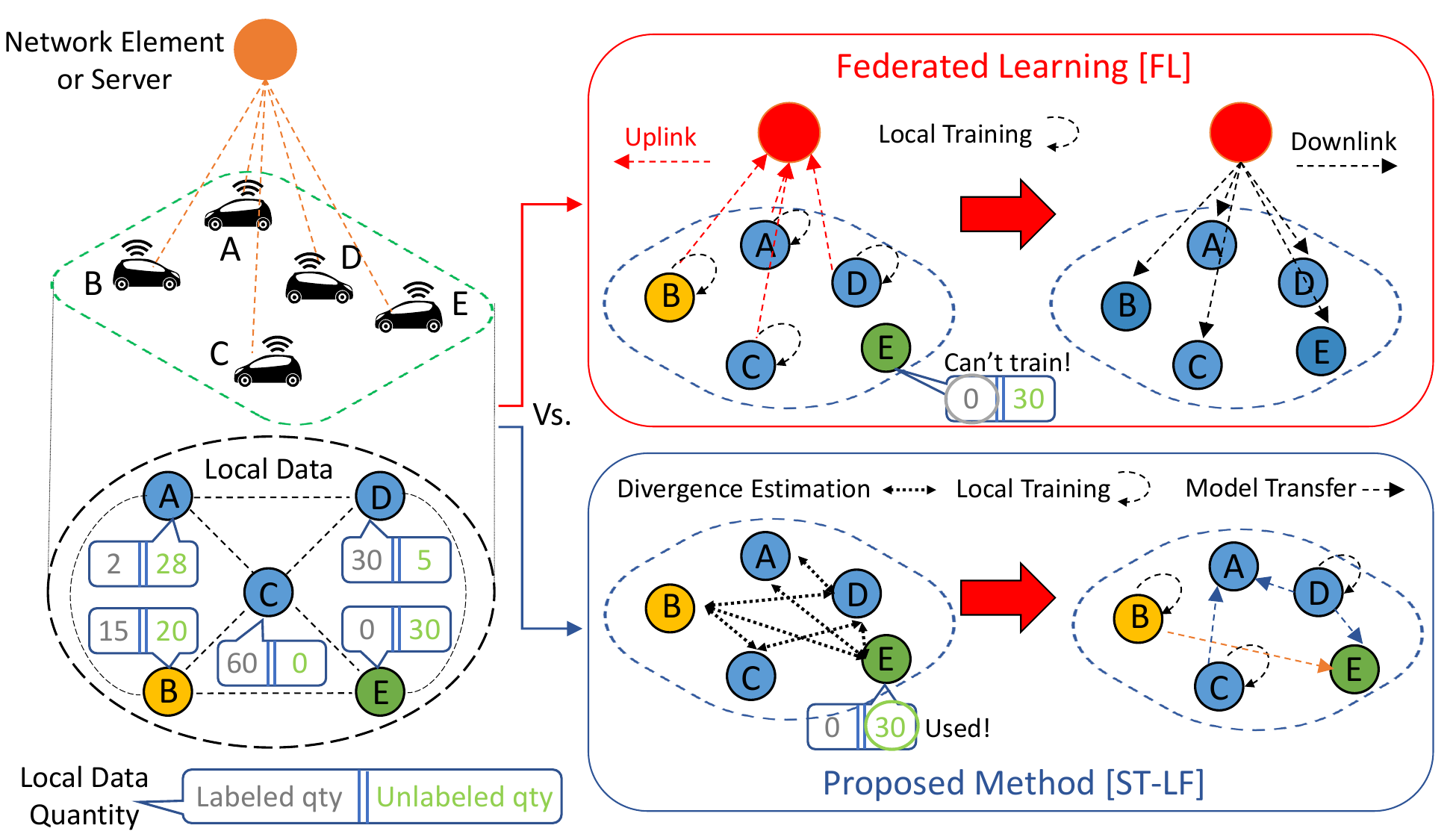}
    \caption{Small motivating example of a network of 5 smart cars. Only cars B, C, and D have meaningful amounts of labeled data, with cars A and E containing very few or no labeled data. Using a server, FL combines ML models from devices with labeled data, yielding a global model heavily biased for the ``blue" domain. Meanwhile, ST-LF uses unlabeled data to estimate pair-wise divergences, then determines source/target selection and source-to-target link formation, leading to individualized ML models without a server.}
    \label{fig:toy_ex}
    \vspace{-5mm}
\end{figure}

\subsection{Motivating Toy Example}
{
Consider Fig.~\ref{fig:toy_ex}, where a network of five smart cars is aiming to collaboratively construct object detection classifiers. Each car gathers data from its local environment, which may contain variations on the same type of object (e.g., images taken in rainy vs sunny environments) and unique sets of objects altogether (e.g., trees vs cacti in the landscape).  
We use the phrase ``data domain" to describe this data-related heterogeneity across network devices, with a colored circle on the left hand side of Fig.~\ref{fig:toy_ex} indicating similarity (in a statistical sense) between local data domains. 
For example, cars A, C, and D all have blue circles, meaning they share very similar data domains. 
Furthermore, each car has a unique local ratio of labeled and unlabeled data.

With traditional FL methods, only cars with labeled data will contribute to ML model training.  
As a result, cars A through D will train ML models while all cars, A through E, will receive the global ML model.  
Since there are more data and devices from the ``blue" domain, the global ML model is biased to favor the ``blue" domain, to the detriment of ``yellow'' (15 labeled data) and ``green'' (0 labeled data) domains.

With our proposed method (i.e., ST-LF), information in \textit{both labeled and unlabeled} data is employed to estimate divergence among device pairs. Using this together with information on resource availability, ST-LF subsequently optimizes the set of source and target cars. In the example of Fig.~\ref{fig:toy_ex}, cars B, C, and D become the source cars, and cars A and E are the target cars. 
Car A's labeled dataset is small, but it has a similar ``blue" data domain with cars C and D, which have much larger labeled datasets. Therefore, car A is likely well-represented by the local models constructed on similar datasets at cars C and D. 
These source cars will locally perform the ML training. 
Then, sources and targets will be matched together based on ST-LF's link formation output. 
In this example, the network can combine the ML models at cars C and D to yield an ML model for car A, as all three cars belong to the ``blue" domain. 
Similarly, combining the ML models at cars B and D can yield a specially-tailored ML model for car E, as combining ``yellow" and ``blue" domains yields a ``green" domain. 
Our optimization will further consider the tradeoff between wireless energy efficiency and ML quality improvement in determining link formation for the network. 
}

\begin{figure*}[t]
    \centering
    \includegraphics[width=0.92\textwidth]{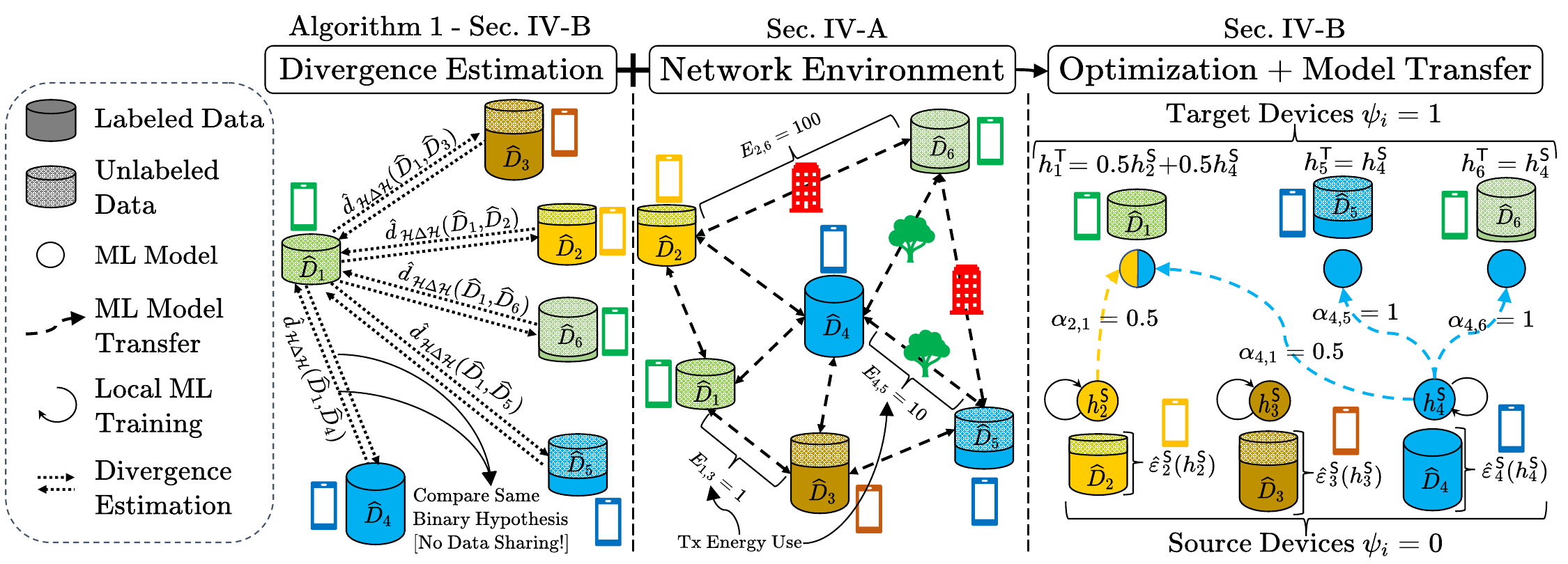}
    \caption{Overview of our ST-LF methodology, where each color represents a different domain. 
    {On the left, ST-LF first determines empirical distribution divergences among device pairs through comparison of a binary domain hypothesis/classifier (visualized for device 1). In the middle, ST-LF extracts information about the network environment such as communication energy costs indicted by ``Tx Energy Use". 
    Finally, on the right, ST-LF uses these measurements in an optimization problem to determine optimal source/target classification $\boldsymbol{\psi}$ and combination weights $\boldsymbol{\alpha}$ with respect to both expected ML model performance and network energy consumption.} }
    \label{fig:intro_model}
     \vspace{-5mm}
\end{figure*}

\subsection{Outline and Summary of Contributions}
Structurally, we first review relevant work in Sec.~\ref{sec:related_work}, and present some important theoretical preliminaries in Sec.~\ref{ssec:relevant_theory}. 
Then, we develop our ST-LF methodology in Sec.~\ref{sec:st-lf-all}, 
and experimentally characterize our formulation itself and its performance relative to several baselines both in Sec.~\ref{sec:experiments}.
We summarize our key contributions as follows: 

\begin{itemize}[leftmargin=3.5mm]
    \item 
    \textbf{Formulation of ST-LF} (Sec.~\ref{sec:st-lf-all}):
    We study a novel yet natural problem in decentralized FL -- only a subset of devices have high quality labeled data.
    We develop a methodology, ST-LF (overview in Fig.~\ref{fig:intro_model}), to address source/target classification, link formation, and communication efficiency requirements simultaneously. 
    \item 
    \textbf{Development of measurable theoretical bounds} (Sec.~\ref{ssec:theoretical_results}):
    We develop a measurable multi-source generalization error bound that captures the impact of combining source hypotheses at targets, which we use to formulate model transfer between multiple sources to multiple targets. 
    {\color{black}Subsequently, we propose an algorithm to estimate source-to-target distribution divergences based on hypothesis comparisons, which offers similar data privacy advantages to those found in FL by eliminating the need to transmit devices' raw data over the network.}
    \item 
    \textbf{Development of optimization methodology for ST-LF} (Sec.~\ref{ssec:init_formulation}):
    As a part of ST-LF, we formulate a novel optimization problem which jointly optimizes our generalization error bound, source-to-target model transfer, and communication efficiency. We show that source/target classification and source-to-target link formation in multi-source to multi-target model transfer scenarios can be classified as a mixed-integer signomial program, a class of NP-hard problems. We subsequently propose a tractable solution using posynomial approximation-based techniques. 
    {The proposed optimization transformation techniques, given their generality and versatility, have a broader range of applicability to federated domain adaptation problems.} 
    \item \textbf{Demonstrate ML performance and energy improvements from ST-LF} (Sec.~\ref{sec:experiments}):
    We experimentally demonstrate the superior performance of ST-FL's source/target classification and link formation relative to baselines from literature. 
    Our experiments reveal the importance of jointly optimizing learning, model transfer, and communication. 
\end{itemize}


\vspace{-1mm}
\section{Related Work}
\label{sec:related_work}

\noindent\textbf{Federated Learning.}
Many methods in FL have been proposed to reduce the impact of non-i.i.d. gathered data on ML model training~\cite{wang2020optimizing,tang2021incentive}. In particular,  multi-task FL~\cite{smith2017federated} learns multiple models, personalized FL~\cite{fallah2020personalized,dinh2020personalized,wang2022uav} adjusts a global model to fit individual devices, and works such as~\cite{wang2021device,zhao2018federated} propose reducing statistical divergence by sharing data via social agreements. 
However, these works rely on all devices having fully labeled data, which poses difficulties in applications where networks of devices only have partially labeled data. 
In the extreme case of fully unlabeled data, methods for unsupervised FL have been proposed, e.g.,~\cite{lu2021federated} using surrogate labels or~\cite{li2021model} which proposes federated contrastive learning. 
We study the ``intermediate'' problem where devices have varying levels of labeled data in FL, and develop novel methods to optimize the transfer of models trained at devices with high quality labeled data, called sources, to those devices without, called targets. 

In our ST-LF methodology, the weighted contributions of source domain ML models are combined at the target domains. This makes our methodology a form of decentralized FL~\cite{ye2022decentralized,roy2019braintorrent}, which exploits the proliferation of device-to-device wireless communications at the edge~\cite{wang2023towards,wang2023potent}.
In this regard, existing works have considered jointly optimizing expected ML model performance and the communication resource consumption that is required over these distributed links~\cite{zhang2021federated,wang2021network,avdiukhin2021federated,hosseinalipour2022parallel}. 
However, these works address centralized settings where the objective is to minimize server-to-device communications. 
We instead consider how expected communication costs among D2D links can influence both source/target device selection and subsequent source-to-target selection in decentralized networks with unlabeled data. 


\textbf{Domain Adaptation.} Domain adaptation~\cite{motiian2017unified,wang2018deep,saito2019semi} is concerned with the transferal of an ML model from a source to a target domain. 
Our problem is a form of unsupervised domain adaptation~\cite{ganin2016domain,sun2016deep} where the target domain is entirely reliant on the source domain.
Whereas most unsupervised domain adaptation methods~\cite{ben2010theory,ganin2016domain,zhao2019learning} assume the labeled and unlabeled data are in a single location, we consider the federated setting where data is distributed across different devices.
To this end,~\cite{peng2019federated} develops an adversarial-based method for distributed domain adaption where a single target and a set of sources is predetermined. In particular, the adversarial learning approach of~\cite{peng2019federated} has inspired several extensions~\cite{zhang2023federated,tran2023personalized} all involving generative adversarial networks and other adversarial methods to perform federated domain adaptation. 
Another recent method~\cite{yao2022federated} entrusts a central server with data from the edge devices in a federated network, which, while effective, has some difficulty lending itself to decentralized FL settings and does lead to some minor data privacy concerns. 
Existing literature~\cite{peng2019federated,yao2022federated,zhao2019learning,zhang2023federated} generally assumes that the source and target devices are known in the network apriori, and has yet to consider the source/target device classification problem by itself when devices may vary with respect to their local ratio and quantity of labeled/unlabeled data. 
Our methodology, enabled by our measurable multi-source generalization error bound, is the first to jointly consider source/target determination, link formation optimization, and communication efficiency for multi-source to multi-target domain adaptation.

\section{Preliminaries}
\label{sec:prelim}

\subsection{Multi-Source/Multi-Target Federated Domain Adaptation}
\label{ssec:prob_defn}
We consider a set $\mathcal{N}$ of devices in a decentralized FL setting (see Fig.~\ref{fig:intro_model}). We aim to partition $\mathcal{N}$ into \textit{source} devices $\mathcal{S}$ and \textit{target} devices $\mathcal{T}$.
Each source $s \in \mathcal{S}$ trains a hypothesis $h^{\mathsf{S}}_s: \mathcal{X} \rightarrow \mathcal{Y}$, $h^{\mathsf{S}}_s \in \mathcal{H}$, where $\mathcal{H}$ denotes the hypotheses space, $\mathcal{X}$ is the input space of the data, and $\mathcal{Y}$ is the classification output. 
We consider $\mathcal{H}$ as a binary hypothesis space for our theoretical analysis, i.e., $\mathcal{Y} \in \{0,1\}$, as in~\cite{ben2010theory,zhao2019learning}. 
Instead of locally training, each target device $t \in \mathcal{T}$ has a hypothesis $h^{\mathsf{T}}_t$ formed based on a weighted combination of source hypotheses, i.e.,
$h^{\mathsf{T}}_t = \sum_{s \in \mathcal{S}} \alpha_{s,t} h^{\mathsf{S}}_s$,
where $\alpha_{s,t} \geq 0$ is the combination weight from source $s$ to target $t$, and $\sum_{s \in \mathcal{S}} \alpha_{s,t} = 1$.


We define the \textit{domain} of source $s$ as $\boldsymbol{\mathcal{D}}^{\mathsf{S}}_s = \langle \mathcal{D}^{\mathsf{S}}_s, f^{\mathsf{S}}_s \rangle$ and target $t$ as $\boldsymbol{\mathcal{D}}^{\mathsf{T}}_t = \langle \mathcal{D}^{\mathsf{T}}_t, f^{\mathsf{T}}_t \rangle$, where $\mathcal{D}^{\mathsf{S}}_s$, $\mathcal{D}^{\mathsf{T}}_t$ are data distributions and $f^{\mathsf{S}}_s$, $f^{\mathsf{T}}_t: \mathcal{X} \rightarrow \mathcal{Y}$ are ground-truth labeling functions at $s$ and $t$, respectively. In general, we expect the $\mathcal{D}^{\mathsf{S}}_s$, $\mathcal{D}^{\mathsf{T}}_t$ to be non-i.i.d.
We denote the true error induced by a source hypothesis as 
\begin{equation} \label{eq:def_true_serror}
    \varepsilon^{\mathsf{S}}_s(h^{\mathsf{S}}_s) = \mathbb{E}_{x \sim \mathcal{D}^{\mathsf{S}}_s}[\vert h^{\mathsf{S}}_s(x) - f^{\mathsf{S}}_s(x)\vert],
\end{equation}
which is the expected deviation of source hypothesis $h^{\mathsf{S}}_s$ from $f^{\mathsf{S}}_s$ over $\mathcal{D}^{\mathsf{S}}_s$. 
Similarly, target devices $t \in \mathcal{T}$ have true error:
\begin{equation} \label{eq:def_true_terror}
    \varepsilon^{\mathsf{T}}_t(h^{\mathsf{T}}_t) = \mathbb{E}_{x \sim \mathcal{D}^{\mathsf{T}}_t}[\vert h^{\mathsf{T}}_t(x) - f^{\mathsf{T}}_t(x)\vert].
\end{equation}
We also use $\varepsilon^{\mathsf{S}}_s(\cdot,\cdot)$ and $\varepsilon^{\mathsf{T}}_t(\cdot,\cdot)$ to compare two hypotheses, e.g., $\varepsilon^{\mathsf{S}}_s(h_1,h_2)$ denotes the true hypothesis divergence error for $h_1,h_2 \in \mathcal{H}$ over $\mathcal{D}^{\mathsf{S}}_{s}$.

\textit{Multi-source to multi-target federated domain adaptation} is the process of transferring hypotheses trained on source domains to target domains to minimize the target domains' true errors. 
%
However, the true error measures $\varepsilon^{\mathsf{S}}_s(\cdot), \varepsilon^{\mathsf{T}}_t(\cdot)$ depend on the underlying data distributions, which are unknown.
We thus use the source empirical error:
\begin{equation} \label{eq:empirical_serror}
\widehat{\varepsilon}^{\mathsf{S}}_s(h^{\mathsf{S}}_s) = \sum_{x \in \widehat{\mathcal{D}}^{\mathsf{S}}_s} \frac{\vert h^{\mathsf{S}}_s(x) - f^{\mathsf{S}}_s(x)\vert} {\widehat{D}^{\mathsf{S}}_s}
\end{equation}
for sources $s \in \mathcal{S}$, where $\widehat{\mathcal{D}}^{\mathsf{S}}_s$ is the empirical dataset at $s$.
{\color{black}In general, sets will be denoted with caligraphic font, e.g., $\mathcal{X}$, and non-caligraphic gives their cardinality, e.g., $X = |\mathcal{X}|$. The $\wedge$ symbol above a variable denotes an empirical result (i.e., the quantity is evaluated over an empirical dataset).}
While the exact ground-truth labeling function $f^{\mathsf{S}}_s$ for $s\in\mathcal{S}$ is unknown, we can determine $f^{\mathsf{S}}_s(x)$ when $x$ is some labeled datum. When $x$ is unlabeled, we treat $\vert h^{\mathsf{S}}_s(x) - f^{\mathsf{S}}_s(x)\vert$ as $1$.\footnote{{\color{black}While the expected error on an unlabeled datum without prior information would be $0.5$, this may not hold as an upper bound for empirical source errors.}} In this manner, the empirical error at a partially labeled device $s$ can be adjusted to also consider the unlabeled data at $s$. 
While target devices $t \in \mathcal{T}$ may also have partially labeled datasets, for mathematical analysis, once classified as targets, target devices are assumed to have no labeled data. Consequently, their empirical target error $\widehat{\varepsilon}^{\mathsf{T}}_t(h^{\mathsf{T}}_t)$ cannot be measured as $f^{\mathsf{T}}_t(x)$ is unknown $\forall x$. 
Instead, empirical hypothesis difference errors are considered at the targets, which we define formally for each target $t$ as:
\begin{equation} \label{eq:target_div_error}
    \widehat{\varepsilon}^{\mathsf{T}}_t(h_1,h_2) =  \sum_{x \in \widehat{\mathcal{D}}^{\mathsf{T}}_t}  \frac{\vert h_1(x) - h_2(x)\vert}{\widehat{D}^{T}_t},
\end{equation}
which can be computed given hypotheses $h_1,h_2 \in \mathcal{H}$. We can similarly calculate $\widehat{\varepsilon}^{\mathsf{S}}_s(h_1,h_2)$ at sources $s \in \mathcal{S}$.
In our problem setting, devices are heterogeneous with respect to the fractions of labeled data that they possess. Classifying devices with limited quantities of labeled data as sources could produce hypotheses that generalize poorly to local unlabeled data as well as to any potential target domains.
We thus desire a new analytical framework for classifying devices in multi-source to multi-target federated domain adaptation, which has to-date remained elusive in the literature. 
To this end, we first review some relevant theory for single-source to single-target domain adaptation in the next section. We then develop our framework in Sec.~\ref{ssec:theoretical_results}.



\subsection{Theoretical Background} 
\label{ssec:relevant_theory}
Works in single-source to single-target domain adaptation~\cite{ben2010theory,zhao2019learning} have proposed analyzing the generalization error of a source hypothesis to a target domain in terms of a \textit{domain divergence} measure: 
\vspace{-1mm}
\begin{definition} ($\mathcal{H}$-divergence) Let $\mathcal{H}$ be the hypothesis class on input space $\mathcal{X}$, and $\mathcal{A}_{\mathcal{H}}$ denote the collection of subsets of $\mathcal{X}$ that form the support of a particular hypothesis in $\mathcal{H}$. The $\mathcal{H}$-divergence between two data distributions $\mathcal{D}$ and $\mathcal{D}^{'}$ is defined as 
$
    d_{\mathcal{H}}(\mathcal{D},\mathcal{D}^{'}) = 2 \sup_{A \in \mathcal{A}_{\mathcal{H}}} \vert Pr_{\mathcal{D}}(A) - Pr_{\mathcal{D}^{'}}(A) \vert,
$
where $\sup$ is the supremum over the subsets and $Pr_{\mathcal{D}}(A)$ is the probability that subset $A$ is in $\mathcal{D}$.
\label{def:div}
\end{definition}
\vspace{-2mm}
A hypothesis class $\mathcal{H}$ induces its own symmetric difference hypothesis space $\mathcal{H}\Delta\mathcal{H} := \{h(x) \oplus h^{'}(x) \vert h, h^{'} \in \mathcal{H} \}$, where $\oplus$ is the XOR operation. The $\mathcal{H}\Delta\mathcal{H}$ space has an associated divergence $d_{\mathcal{H}\Delta\mathcal{H}}$ defined as in Definition~\ref{def:div} but with $\mathcal{A} \in \mathcal{A}_{\mathcal{H}\Delta\mathcal{H}}$ instead. 
Denoting the optimal joint hypothesis for a single source $s$ and target $t$ as $h^{*}_{s,t} := \argmin_{h \in \mathcal{H}} \{ \varepsilon^{\mathsf{S}}_s(h) + \varepsilon^{\mathsf{T}}_t(h) \}$ and the corresponding minimum joint error as $\lambda^{*}_{s,t} := \varepsilon^{\mathsf{S}}_s(h^*) + \varepsilon^{\mathsf{T}}_t(h^*)$,
~\cite{ben2010theory} proved a generalization bound on the target error $ \varepsilon^{\mathsf{T}}_t(h)$ in terms of the empirical source error and the empirical $\widehat{d}_{\mathcal{H}\Delta\mathcal{H}}$ divergence:

\vspace{-.1mm}

\begin{restatable}{theorem}{bd_2010} \label{thm:bd_2010} Let $\mathcal{H}$ be a hypothesis space of Vapnik–Chervonenkis (VC) dimension $d$, $s$ be a source domain, $t$ be a target domain, and $\widehat{\mathcal{D}}^{\mathsf{S}}_s$, $\widehat{\mathcal{D}}^{\mathsf{T}}_t$ be the empirical distributions induced by samples of size $n$ drawn from $\mathcal{D}^{\mathsf{S}}_s$ and $\mathcal{D}^{\mathsf{T}}_t$. Then, with probability of at least $1-\delta$ over the choice of samples,
\vspace{-2mm}

{\normalsize
\begin{align}\label{eq:bd_2010}
    &\varepsilon^{\mathsf{T}}_t(h) \leq \widehat{\varepsilon}^{\mathsf{S}}_s(h) + \frac{\widehat{d}_{\mathcal{H} \Delta \mathcal{H}} (\widehat{\mathcal{D}}^{\mathsf{S}}_s,\widehat{\mathcal{D}}^{\mathsf{T}}_t)}{2} \\ \nonumber
    & + 4 \sqrt{\frac{2d \log(2n) + \log(4/\delta)}{n} } + \lambda^{*}_{s,t}, \forall h \in \mathcal{H}.
\end{align}}
\end{restatable}
Existing literature~\cite{peng2019federated,zhao2018adversarial}, which has extended Theorem~\ref{thm:bd_2010} to multi-source domain adaptation, has relied on identifying an optimal hypothesis $h^{*}_{s,t}$ for every source $s \in \mathcal{S}$ to a single target $t$ in order to compute a combined minimum error. 
However, finding $h^{*}_{s,t}$ and $\lambda^{*}_{s,t}$ requires optimizing the error over an entire hypothesis space $\mathcal{H}$ for both source and target domains, which is impractical in our setting given the inability to centralize data in many FL applications. Furthermore, while these bounds consider the worst-case where the errors across sources are independent, in practice they may exhibit dependencies which can be exploited to obtain tighter bounds.
Motivated by this, in Sec.~\ref{ssec:theoretical_results}, we develop a multi-source generalization error bound which considers the \textit{joint} effect of sources, eliminating the need for $\lambda^{*}_{s,t}$ terms from~\eqref{eq:bd_2010}. Subsequently, we formulate our link formation optimization problem that jointly minimizes source and target errors. 

\section{Source-Target Determination and Link Formation in FL (ST-LF)}
\label{sec:st-lf-all}

\vspace{-.1mm}
\subsection{Generalization Error Characterization} \label{ssec:theoretical_results}

One of our central contributions is developing a multi-source generalization error bound for the true target error $\varepsilon^{\mathsf{T}}_t(h^{\mathsf{T}}_t)$ that consists of optimizable/controllable terms.
In the process to obtain such a bound, we first bound the true target error $\varepsilon^{\mathsf{T}}_t(h^{\mathsf{T}}_t)$ as a combination of terms that compare the domain at target $t$ to the domains at sources $s \in \mathcal{S}$. This makes the bound a function of controllable combination weights $\alpha_{s,t}$ and reveals their impact on the performance. 


\begin{restatable}{theorem}{ourgenbound} \label{thm:upper_bound_tte} (Upper Bound for True Target Error) 
Given a set of sources $\mathcal{S}$ and a weighted target hypothesis $h^{\mathsf{T}}_t = \sum_{s \in \mathcal{S}} \alpha_{s,t} h^{\mathsf{S}}_s$, where $0 \leq \alpha_{s,t} \leq 1~\forall s,t,$ and $\sum_{s \in \mathcal{S}} \alpha_{s,t} = 1$, for a target $t \in \mathcal{T}$, the true target error is bounded as follows: 

\vspace{-2mm}
{\normalsize
\begin{align} 
     \varepsilon^{\mathsf{T}}_t (h^{\mathsf{T}}_t ) \leq &
    \sum_{s \in \mathcal{S}} \hspace{-0.1mm} \alpha_{s,\hspace{-0.1mm}t} \hspace{-0.1mm} \bigg[ \hspace{-0.1mm}
    \underbrace{\varepsilon^{\mathsf{S}}_s\hspace{-0.1mm}(h^{\mathsf{S}}_s)}_{(i)} 
    \hspace{-0.1mm} + \hspace{-0.1mm}\nonumber \\
    & \underbrace{\varepsilon^{\mathsf{T}}_t\hspace{-0.1mm}(f^{\mathsf{S}}_s)}_{(ii)}
    \hspace{-0.1mm} + \hspace{-0.1mm}
    \underbrace{\frac{d_{\mathcal{H}\Delta \mathcal{H}}\hspace{-0.1mm}(\mathcal{D}^{\mathsf{T}}_t,\hspace{-0.1mm}\mathcal{D}^{\mathsf{S}}_s)}{2}}_{(iii)} 
    \hspace{-0.1mm} + \hspace{-0.1mm}
    \underbrace{\varepsilon^{\mathsf{T}}_t\hspace{-0.1mm}(h^{\mathsf{T}}_t\hspace{-0.1mm},\hspace{-0.1mm}h^{\mathsf{S}}_s)}_{(iv)} \bigg].\label{eq:thm1_main_result} \hspace{-2mm}
\end{align}}
\end{restatable}
\vspace{-1.5mm}
The proof of Theorem~\ref{thm:upper_bound_tte} is given in Appendix~\ref{app_ssec:thm2}; the proofs of all subsequent results are in their appropriate appendices. 

Theorem~\ref{thm:upper_bound_tte} elucidates a set of theoretical criteria for efficient model transfer: (i) a source hypothesis should have small error $\varepsilon^{\mathsf{S}}_s(h^{\mathsf{S}}_s)$, (ii) the ground-truth labeling functions at the source and target should be similar as measured by $\varepsilon^{\mathsf{T}}_t(f^{\mathsf{S}}_s,f^{\mathsf{T}}_t)$, (iii) the underlying data distributions at the sources and target should be similar as measured by $\frac{1}{2}d_{\mathcal{H}\Delta \mathcal{H}}(\mathcal{D}^{\mathsf{T}}_t,\mathcal{D}^{\mathsf{S}}_s)$, and 
(iv) a source hypothesis should lead to a small hypothesis combination noise $\varepsilon^{\mathsf{T}}_t(h^{\mathsf{T}}_t,h^{\mathsf{S}}_s)$. 
Given a target $t$, the bound in~\eqref{eq:thm1_main_result} favors a larger combination weight $\alpha_{s,t}$ between a source-target pair $(s,t)$ (i.e., model transfer between the two) if the hypothesis/domain at a source $s$ satisfies the four aforementioned criteria.




The bound in~\eqref{eq:thm1_main_result} contains terms defined based on the underlying data distributions, e.g., $\mathcal{D}^{\mathsf{T}}_t$, that cannot be measured in practice.  
We thus aim to transform terms $(i)$, $(iii)$, and $(iv)$ in~\eqref{eq:thm1_main_result} to measurable quantities that can be computed in practical systems.\footnote{Term $(ii)$ in~\eqref{eq:thm1_main_result} is based on ground-truth labeling functions and thus is a constant.}
To this end, we exploit the empirical Rademacher complexity and transform \textit{true} error measures (e.g., $\varepsilon^{\mathsf{T}}_t$) defined on unknown data distributions to \textit{empirical} error measures (e.g., $\widehat{\varepsilon}^{\mathsf{T}}_T$), which are computable.

\begin{definition} (Empirical Rademacher Complexity). Let $\mathcal{H}$ be a family of functions mapping from input domain $\mathcal{X}$ to an interval $[a,b]$ and $\boldsymbol{Q} = \{ x_i \}_{i=1}^{n}$ be a fixed sample of size $n$ with elements in $\mathcal{X}$. Then, the empirical Rademacher complexity of $\mathcal{H}$ with respect to the sample $\boldsymbol{Q}$ is defined as \cite{bartlett2002rademacher}:

\vspace{-2.5mm}
{\small
\begin{equation} \label{eq:rad}
    Rad_{\boldsymbol{Q}}(\mathcal{H}) = \mathbb{E}_{\sigma} \bigg[ \sup_{h \in \mathcal{H}} \frac{1}{n} \sum_{i=1}^{n} \sigma_i h(x_i) \bigg],
\end{equation}}
where $\boldsymbol{\sigma} = \{ \sigma_i \}_{i=1}^{n}$ and $\sigma_i$ are i.i.d. random variables taking values in $\{ +1,-1 \}$.
\end{definition}

Empirical Rademacher complexity is a measure of the complexity of a hypothesis space $\mathcal{H}$ given a set of data $\bm{Q}$. 
Utilizing this measure, we show that true hypothesis noise (i.e., term $(iv)$ in \eqref{eq:thm1_main_result}) can be approximated via empirical errors.

\begin{restatable}{lemma}{hcombination} \label{lemma1} (Upper Bound for Hypothesis Combination Noise) Let $\mathcal{H}$ be a binary hypothesis space. For arbitrary $0 < \delta < 1$, the following bound holds $\forall h \in \mathcal{H}$ with probability of at least $1-\delta$: 
\vspace{-1mm}
\begin{equation} \label{eq:lemma1}
    \varepsilon^{\mathsf{T}}_t(h^{\mathsf{T}}_t, h^{\mathsf{S}}_s) 
    \leq \widehat{\varepsilon}^{\mathsf{T}}_t(h^{\mathsf{T}}_t,h^{\mathsf{S}}_s) 
    + 4 Rad_{\widehat{\mathcal{D}}^{\mathsf{T}}_t}(\mathcal{H}) + 3 \sqrt{\frac{\log(2/\delta)}{2\widehat{D}^{\mathsf{T}}_t}}.
\end{equation}
\end{restatable}
\vspace{-2mm}
Lemma~\ref{lemma1} obtains the true hypothesis comparison error $\varepsilon^{\mathsf{T}}_t(h^{\mathsf{T}}_t,h^{\mathsf{S}}_s)$ between hypotheses $h^{\mathsf{T}}_t$ and $h^{\mathsf{S}}_s$ with respect to the data distribution $\mathcal{D}^{\mathsf{T}}_t$ at target $t$. 
Similar to the above procedure, for term $(i)$ in \eqref{eq:thm1_main_result}, we can obtain 
\begin{equation}
\varepsilon^{\mathsf{S}}_s(h^{\mathsf{S}}_s) \leq \widehat{\varepsilon}^{\mathsf{S}}_s(h^{\mathsf{S}}_s) + 2 Rad_{\widehat{\mathcal{D}}^{\mathsf{S}}_s}(\mathcal{H}) + 3\sqrt{\log(2/\delta) / (2\widehat{D}^{S}_s)}
\end{equation}
which resembles the results in~\cite{zhao2018adversarial}. 


The terms in \eqref{eq:lemma1} are measurable quantities. 
{\color{black}Specifically, the empirical Rademacher complexity $Rad_{\widehat{\mathcal{D}}^{\mathsf{T}}_t}(\mathcal{H})$ can be bounded as $Rad_{\widehat{\mathcal{D}}^{\mathsf{T}}_t}(\mathcal{H})\leq \sqrt{2\log(2)}$ via Massart's Lemma (Lemma~\ref{lemma:massart_lemma} in Appendix~\ref{app_ssec:massart_lemma}). 
This is a general bound for binary hypothesis spaces $\mathcal{H}$ as it accounts for the worst-case, in which $\mathcal{H}$ shatters $\widehat{\mathcal{D}}^{\mathsf{T}}_t$.}
We next revisit Theorem~\ref{thm:upper_bound_tte} to obtain a measurable bound on the true target error: 
\begin{restatable}{corollary}{boundrad} \label{coro:gen_error_bound_rademacher} (Measurable Error Bound for True Target Error) Let $\mathcal{H}$ be a hypothesis class, and $\{\widehat{\mathcal{D}}^{\mathsf{S}}_s \}_{s \in \mathcal{S}}$ and $\widehat{\mathcal{D}}^{\mathsf{T}}_t$ be empirical distributions of sources and target domains. 
For arbitrary $0 < \delta < 1$, the following bound holds $\forall h \in \mathcal{H}$ with probability of at least $1-\delta$: 
\vspace{-1.5mm}

{\color{black}
{\normalsize
\begin{align} \label{eq:coro_gen_error_bound_rademacher}
    &\varepsilon^{\mathsf{T}}_t(h^{\mathsf{T}}_t) \hspace{-0.5mm} 
    \leq \hspace{-0.5mm} 
    \sum_{s \in \mathcal{S}} \hspace{-0.8mm} \alpha_{s,t} \hspace{-0.6mm} \Bigg[ \hspace{-0.6mm}
    \underbrace{\widehat{\varepsilon}^{\mathsf{S}}_s(h^{\mathsf{S}}_s) \hspace{-0.8mm} + \hspace{-0.8mm}
    \frac{\widehat{d}_{\mathcal{H} \Delta \mathcal{H}} (\widehat{\mathcal{D}}^{\mathsf{T}}_t \hspace{-0.7mm} , \hspace{-0.7mm} \widehat{\mathcal{D}}^{\mathsf{S}}_s)}{2} \hspace{-0.8mm} 
    + \hspace{-0.8mm} 
    \varepsilon^{\mathsf{T}}_t \hspace{-0.5mm}(f^{\mathsf{S}}_s) \hspace{-0.6mm} + \hspace{-0.6mm} 
    \widehat{\varepsilon}^{\mathsf{T}}_t \hspace{-0.5mm} (h^{\mathsf{T}}_t\hspace{-0.8mm} , \hspace{-0.5mm} h^{\mathsf{S}}_s) }_{(a)} 
    \\ \nonumber 
    & + \hspace{-1mm} 
    10\sqrt{2\log(2)}
    + 
    \underbrace{
    6\Bigg( \hspace{-1.2mm} 
    \left( \hspace{-0.6mm} {\frac{\log(2/\delta)}{2\widehat{D}^{\mathsf{T}}_t}} \hspace{-0.6mm} \right)^{\hspace{-1mm}\frac{1}{2}}
    \hspace{-2mm} + \hspace{-1.2mm} 
    \left( \hspace{-0.6mm}{\frac{\log(2/\delta)}{2\widehat{D}^{\mathsf{S}}_s}} \hspace{-0.6mm} \right)^{\hspace{-1mm} \frac{1}{2}} \hspace{-0.6mm}  }_{(b)} \hspace{-.5mm}\Bigg) \hspace{-0.6mm} \Bigg]. 
\end{align}
} 
} 
\end{restatable}
\vspace{-1mm}
We group the terms in~\eqref{eq:coro_gen_error_bound_rademacher} into two categories. The first, $(a)$, captures the impact of source hypotheses, data distributions, and ground-truth labeling functions (i.e., the empirical source error, the empirical divergence, the ground-truth labeling function difference, and the hypothesis combination noise). 
The second, $(b)$, considers the impact of data quantity at the sources and targets on the  error. 

Unlike existing multi-source generalization error bounds~\cite{peng2019federated,zhao2018adversarial}, our result in Corollary~\ref{coro:gen_error_bound_rademacher} does not require pairwise hypothesis optimization for $h^{*}_{s,t}$ to minimize pairwise error $\lambda^{*}_{s,t}$, which, as discussed in Sec.~\ref{ssec:relevant_theory}, enables us to consider the joint impacts of sources $s \in \mathcal{S}$ used in model transfer to a target $t$. 
Using this important property, we next formulate the multi-source to multi-target federated domain adaptation problem as a link formation optimization.



\subsection{ST-LF Optimization Formulation and Solver}
\label{ssec:init_formulation}
ST-LF aims to transfer source-trained ML models to targets in order to maximize the ML performance across all devices while efficiently using network communication resources. 
To this end, we should first determine the source-target classification across devices ${\boldsymbol{\psi}}=\{\psi_i\}_{i\in\mathcal{N}}$, with $\psi_i = 0$ if device $i$ is a source and $\psi_i = 1$ otherwise, since that is not known \textit{a priori}. Then, the combination weights (interpreted as link/edge weights) between different source-target pairs $\boldsymbol{\alpha}=\{\alpha_{i,j}\}_{i,j\in\mathcal{N}}$ should be designed, as visualized in Fig.~\ref{fig:intro_model}.
Formally, we pose ST-LF as the following optimization problem:
\vspace{-.5mm}
{\normalsize
\begin{align}
    & \hspace{-3mm} (\boldsymbol{\mathcal{P}}):~\argmin_{\boldsymbol{\alpha},\boldsymbol{\psi}} 
    \phi^{\mathsf{S}} \underbrace{\sum_{i \in \mathcal{N}} (1-\psi_i) S_i}_{(c)} \nonumber \\
    & \hspace{-3mm} + \phi^{\mathsf{T}} \hspace{-1mm} \underbrace{\sum_{j \in \mathcal{N}} \hspace{-0.5mm} \psi_j \hspace{-0.5mm} \sum_{i \in \mathcal{N}} (1-\psi_i) \alpha_{i,j} T_{i,j}}_{(d)} 
    + \phi^{\mathsf{E}} \underbrace{\sum_{i \in \mathcal{N}} \sum_{j \in \mathcal{N}} E_{i,j} (\alpha_{i,j})}_{(e)} \label{eq:obj_fxn_1} \\
    & \textrm{subject to} \nonumber \\
    & h^{\mathsf{T}}_j = \sum_{i \in \mathcal{N}} \alpha_{i,j} (1-\psi_i)\psi_j h^{\mathsf{S}}_i, \forall j \in \mathcal{N} \label{eq:target_hypothesis_def}\\
    & \sum_{i \in \mathcal{N}} \alpha_{i,j} = \psi_j, \forall j \in \mathcal{N} \label{eq:only_targets_receive} \\
    & E_{i,j} (\alpha_{i,j}) = K_{i,j} \frac{\alpha_{i,j}}{\alpha_{i,j} + \epsilon_{E}}, \forall i,j \in \mathcal{N} \label{eq:def_E_i} \\ 
    & 0 \leq \alpha_{i,j} \leq 1, \forall i,j \in \mathcal{N} \label{eq:alpha_limits} \\
    & \psi_i(t) \in \{ 0,1 \}, \forall i,j \in \mathcal{N}. \label{eq:psi_limits}
\end{align}
}
\vspace{-4mm}

\vspace{1mm}
\textbf{Objective of $(\boldsymbol{\mathcal{P}})$.} \eqref{eq:obj_fxn_1} captures the trade-off between source/target classification errors and communication energy consumption from model transfers. In particular, $(c)$ captures the true source error obtained in Lemma~\ref{lemma1}, where 
{\color{black}
\begin{equation}
S_i = \widehat{\varepsilon}^{\mathsf{S}}_i(h^{\mathsf{S}}_i) 
+ 2 \sqrt{2\log(2)} 
+ 3 \sqrt{ \log(2/\delta)/(2\widehat{D}^{\mathsf{S}}_i)}.
\end{equation}
}
Also, $(d)$ captures the error bound at target domains deduced in Corollary~\ref{coro:gen_error_bound_rademacher} through 
{\normalsize
{\color{black}
\begin{equation}
\begin{aligned}
    & T_{i,j} = \widehat{\varepsilon}^{\mathsf{S}}_i(h^{\mathsf{S}}_i) + 10 \sqrt{2\log(2)} \\
    & + \varepsilon^{\mathsf{T}}_j(f^{\mathsf{T}}_j,f^{\mathsf{S}}_i) + \frac{1}{2}\widehat{d}_{\mathcal{H}\Delta\mathcal{H}}(\widehat{\mathcal{D}}^{\mathsf{T}}_j,\widehat{\mathcal{D}}^{\mathsf{S}}_i) + \widehat{\varepsilon}^{\mathsf{T}}_j(h^\mathsf{T}_j,h^{\mathsf{S}}_i) \\
    & + 6 \left( \sqrt{ \log(2/\delta)/(2\widehat{D}^{\mathsf{S}}_i)}
    + \sqrt{ \log(2/\delta)/ (2\widehat{D}^{\mathsf{T}}_j)} \right).
\end{aligned}
\end{equation}}}
Finally, $(e)$ measures the communication energy consumption specified in~\eqref{eq:def_E_i}. 
Minimizing $(c)$ in $(\boldsymbol{\mathcal{P}})$ influences source/target classification $\boldsymbol{\psi}$ as the optimization will aim to classify only those devices with high quality data as sources. 
Terms $(d)$ and $(e)$ primarily determine the combination weights $\boldsymbol{\alpha}$ because the true target errors in term $(d)$ can be minimized by proper selection of $\boldsymbol{\alpha}$ to optimize target hypotheses $h^{\mathsf{T}}_j$, and the communication energy cost in term $(e)$ is only incurred when a source $i$ transmits its hypothesis to a target $j$, i.e., only when $\alpha_{i,j} > 0$ due to~\eqref{eq:def_E_i}. 
We can shift the relative importance of terms $(c)$, $(d)$, and $(e)$ in~\eqref{eq:obj_fxn_1} by adjusting the scaling coefficients $ \phi^{\mathsf{S}},\phi^{\mathsf{T}},\phi^{\mathsf{E}}\geq 0$. 
In the extreme case, $\phi^{\mathsf{S}}=0$ makes all devices sources ($\mathcal{S} = \mathcal{N}$), $\phi^{\mathsf{T}}=0$ makes all devices targets ($\mathcal{T}=\mathcal{N}$), and $\phi^{\mathsf{E}}=0$ means the optimization does not consider communication energy use. 

$(\boldsymbol{\mathcal{P}})$ is the first concrete formulation of link formation in FL domain adaptation enabled via our results in Sec.~\ref{ssec:theoretical_results}. 
{\color{black}While communication energy is chosen as the cost of an active link between devices $i$ and $j$ (similarly to~\cite{tran2019federated,wang2022uav}), the factor $K_{i,j}$ in~\eqref{eq:obj_fxn_1} and~\eqref{eq:def_E_i} could be adjusted on a case-by-case basis to fit various other potential link costs of interest such as latency, privacy constraints, or link formation, as long as these costs do not introduce factors that violate signomial/geometric programming rules that we will employ later in this section.}
{\color{black}Different practical link formation costs can lead $(\boldsymbol{\mathcal{P}})$ to yield alternative model combinations.}

{\color{black}
More generally, we can draw parallels between $(\boldsymbol{\mathcal{P}})$ and link formation problems, which have been well studied in network science. This includes in wireless/IoT networks, where links are established for e.g., throughput maximization~\cite{ning2017social}, and in social networks, where link costs correspond to e.g., privacy/trust measures~\cite{zhang2015integrated}. 
However, existing link formation methodologies are not well-suited for the federated domain adaptation problem, presented in $(\boldsymbol{\mathcal{P}})$, which poses a tight coupling between ML performance and resource efficiency.}

\textbf{Constraints of $(\boldsymbol{\mathcal{P}})$.} 
We replace the original definition of $h^{\mathsf{T}}_j=\sum_{i \in \mathcal{S}} \alpha_{i,j} h^{\mathsf{S}}_i$ (Sec.~\ref{ssec:prob_defn}) with 
\eqref{eq:target_hypothesis_def} by adding the coefficient $(1-\psi_i)\psi_j$ inside the sum. This coefficient ensures that target hypotheses are only reliant on source hypotheses since $(1-\psi_i)\psi_j=1$ if and only if $\psi_i=0$ (i.e., $i$ is a source) and $\psi_j=1$ (i.e., $j$ is a target).
Constraint~\eqref{eq:only_targets_receive} guarantees that only target devices ($\psi_j=1$) receive models (i.e., $\sum_{i\in\mathcal{N}} \alpha_{i,j} =1$) and source devices ($\psi_j=0$) only send models (i.e., $\sum_{i\in\mathcal{N}} \alpha_{i,j}=0$).~\eqref{eq:def_E_i} models when communication energy is incurred between two devices $i$ and $j$, i.e., only if model transfer occurs between them ($\alpha_{i,j} > 0$). 
{\color{black}
In fact, $\alpha_{i,j}$ only has two states: either $\alpha_{i,j}>0$ so source $i$ transmits to target $j$ or $\alpha_{i,j} = 0$, meaning there is no transmission.
}
To capture this in a tractable manner, we approximate the desired behavior via $\frac{\alpha_{i,j}}{\alpha_{i,j}+\epsilon_E}$, where $0<\epsilon_E \ll 1$, in~\eqref{eq:def_E_i}, so that $\frac{\alpha_{i,j}}{\alpha_{i,j}+\epsilon_E}\approx 1$ when $\alpha_{i,j}>0$, and $\frac{\alpha_{i,j}}{\alpha_{i,j}+\epsilon_E}=0$ when $\alpha_{i,j}=0$. 
$K_{i,j}$ is the energy/cost of model transfer from $i$ to $j$, for which we adopt physical layer communication models presented in Sec.~\ref{sec:experiments}. 
{\color{black}Thus, in our work, the energy consumption $K_{i,j}$ is taken as the cost of link formation.}
Finally,~\eqref{eq:alpha_limits} and~\eqref{eq:psi_limits} are the feasibility constraints for offloading ratios and source/target classification respectively.

\begin{algorithm}[t]
 	\caption{Determination of Empirical Divergences}\label{alg:fed_div_est}
 	\SetKwFunction{Union}{Union}\SetKwFunction{FindCompress}{FindCompress}
 	\SetKwInOut{Input}{input}\SetKwInOut{Output}{output}
 	{\small
 	\Input{Starting Domain Classification Hypothesis $h^{'}$, Network Devices $\mathcal{N}$, Local Training Period $T^{d}$, Number of Local 
        Aggregations $\tau^{d}$}}
        {\small
 	\For{Device pair $(i,j)$, $i \in \mathcal{N}$ $j \in \mathcal{N}$, $i \neq j$}{
 	Initialize hypothesis at both devices $i$ and $j$, i.e., $h_i=h_j=h^{'}$. \\
 	Label all $x_d \in \widehat{\mathcal{D}}_i$ as $0$ and all $x_d \in \widehat{\mathcal{D}}_j$ as $1$. \\
 	\For{$\tau_i=1$ to $\tau^{d}$}{
 	Locally train hypotheses at $i$ and $j$ on local data $\widehat{\mathcal{D}}_i$ and $\widehat{\mathcal{D}}_j$ respectively for $T^{d}$ 
        iterations. \\
 	Transfer hypothesis $h_i$ at device $i$ to device $j$, and vice versa. \\
 	Take the average of $h_i$ and $h_j$, i.e., $\bar{h}^{'} = \frac{h_i+h_j}{2}$. \\
 	}
 	Measure classification error of final hypothesis $\bar{h}^{'}$ on data at devices $i$ and $j$. \\
        Transfer error of final hypothesis at device $i$ to device $j$, and vice versa. \\
 	Compute $\widehat{d}_{\mathcal{H}}(\widehat{\mathcal{D}}_j,\widehat{\mathcal{D}}_i)$ using the classification error. \\
 	}
 	Obtained $\widehat{d}_{\mathcal{H}}(\widehat{\mathcal{D}}_j,\widehat{\mathcal{D}}_i)$ for all potential source and target pair combinations, i.e., $\forall i \in \mathcal{N}$ and $ \forall j \in \mathcal{N}$. 
 	}
\end{algorithm}

\textbf{Computing the terms of $({\boldsymbol{\mathcal{P}}})$.}
Before we can solve $({\boldsymbol{\mathcal{P}}})$, the optimization terms must first be computed. In particular, the source error terms embedded within $S_i$ (found in part $(c)$ of~\eqref{eq:obj_fxn_1}) can all be computed directly, for example, the empirical source errors, $\widehat{\varepsilon}^{\mathsf{S}}_i(h^{\mathsf{S}}_i)$, can be locally determined by devices themselves. 
While the empirical error terms, $\widehat{\varepsilon}^{\mathsf{T}}_j$ and $\widehat{\varepsilon}^{\mathsf{S}}_i$, embedded within $T_{i,j}$ can be similarly computed, the remaining terms found in part $(d)$ of~\eqref{eq:obj_fxn_1} require different strategies. 
Firstly, the ground-truth labeling function difference, $\varepsilon^{\mathsf{T}}_{j}(f^{\mathsf{T}}_j,f^{\mathsf{S}}_i)$ cannot be found in practice, as it requires devices to have concrete ground-truth labeling functions. If devices have concrete ground-truth labeling functions, then the network has no need for any ML, and, therefore, we can omit the ground-truth labeling function differences from the computation of $T_{i,j}$. 

{\color{black}Owing to practical difficulties in computing $\widehat{d}_{\mathcal{H}\Delta\mathcal{H}}$ generally, we follow a similar strategy as~\cite{ben2010theory,zhao2018adversarial} in using $\widehat{d}_{\mathcal{H}}$ as an estimate for $\widehat{d}_{\mathcal{H}\Delta\mathcal{H}}$.} 
{\color{black}To compute the empirical $\frac{1}{2}\widehat{d}_{\mathcal{H}}(\widehat{\mathcal{D}}^{\mathsf{T}}_j,\widehat{\mathcal{D}}^{\mathsf{S}}_i)$ divergence, we propose Algorithm~\ref{alg:fed_div_est}, a decentralized peer-to-peer algorithm designed to operate pairwise between devices.}
Since empirical distribution divergences quantify the separability of two domains based on their local data~\cite{ben2010theory}, our algorithm uses a binary classifier at two devices to obtain $\frac{1}{2}\widehat{d}_{\mathcal{H}}(\widehat{\mathcal{D}}^{\mathsf{T}}_j,\widehat{\mathcal{D}}^{\mathsf{S}}_i)$. 
Essentially, algorithm~\ref{alg:fed_div_est} has three key steps: (i) relabel data as $0$ at the source and $1$ at the target, (ii) separately train and combine binary classifiers at the source and target domains, and (iii) use the final binary classifier to find the empirical separability of the two domains. 
{\color{black}Since Algorithm~\ref{alg:fed_div_est} only relies on sharing the parameters of this binary domain classifier, it thus offers similar data privacy advantages found in FL by avoiding raw data sharing among devices.}
Further discussion of Algorithm~\ref{alg:fed_div_est} is provided in Appendix~\ref{app_ssec:divergence_estimation}.

\begin{algorithm}[t] 
 	\caption{Optimization solver for problem~$\bm{\mathcal{P}}$}\label{alg:cent}
 	\label{alg:optimization_iteration}
 	\SetKwFunction{Union}{Union}\SetKwFunction{FindCompress}{FindCompress}
 	\SetKwInOut{Input}{input}\SetKwInOut{Output}{output}
  	{\small
 	\Input{Convergence criterion.}
 	\Output{$\bm{x}^\star, \textrm{Objective of $\bm{\mathcal{P}}$ evaluated at $\bm{x}^\star$} $}
 	 }
 	 {\small
 	 Set the iteration count $\ell=0$.\\
 	 Choose a feasible point $\bm{x}^{[0]}$.\\
 	 Obtain the monomial approximations~\eqref{eq:approx_source_error},\eqref{eq:approx_target_error1},\eqref{eq:approx_target_error2},\eqref{eq:approx_nrg},\eqref{eq:approx_con_plus},\eqref{eq:approx_con_minus} given $\bm{x}^{[\ell]}$.\label{midAlg1}\\
 	 Replace the results in the approximation of Problem~$\bm{\mathcal{P}}$ (see $\bm{\mathcal{P}}'$ in Appendix~\ref{app_sssec:tx_2gp}). \\ 
 	 With logarithmic change of variables, transform the resulting GP problem to a convex problem. \\ 
 	 $\ell=\ell+1$\\
 	 Obtain the solution of the convex problem using current art solvers (e.g., CVXPY~\cite{diamond2016cvxpy}) to determine  $\bm{x}^{[\ell]}$.\label{Alg:Gpconvexste}\\
 	 \If{two consecutive solutions $\bm{x}^{[\ell-1]}$ and $\bm{x}^{[\ell]}$ do not meet the specified convergence criterion}{
 	\textrm{Go to line~\ref{midAlg1} and redo the steps using $\bm{x}^{[\ell]}$.}\\\Else{Set the solution of the iterations as $\bm{x}^{\star}=\bm{x}^{[\ell]}$.\label{Alg:point2}\\}
 	}}
\end{algorithm} 

\textbf{Solution of $({\boldsymbol{\mathcal{P}}})$.}
{\color{black}
In $({\boldsymbol{\mathcal{P}}})$, the optimization variables are tightly coupled together - as any choice of source/target classification $\psi_i$, $\forall i \in \mathcal{N}$, directly influences possible model offloading ratios $\alpha_{i,j}$, $\forall i,j \in \mathcal{N}$. This coupling manifests itself through the multiplication of optimization variables. For example, term $(d)$ in~\eqref{eq:obj_fxn_1} and~\eqref{eq:target_hypothesis_def} contains the multiplication of real variable $\alpha_{i,j}$ and integer variable $\psi_j$. 
Furthermore, these variable combinations often involve (i) negative optimization variables such as the $1-\psi_i$ term found in $(d)$ and~\eqref{eq:target_hypothesis_def}, (ii) equality constraints such as~\eqref{eq:only_targets_receive}, and (iii) the division of sums involving optimization variables such as~\eqref{eq:def_E_i}.  
As a result of these complex combinations of optimization variables, $({\boldsymbol{\mathcal{P}}})$ belongs to a class of mixed-integer signomial programs, which are known to be NP-hard and non-convex~\cite{chiang2005geometric}. 
Since this formulation captures the intricate and coupled process from source/target identification to model ratio offloading, this complexity is expected. 

We now develop a tractable solution for $({\boldsymbol{\mathcal{P}}})$ based on a set of modifications and approximations for the objective function~\eqref{eq:obj_fxn_1} and constraints~\eqref{eq:target_hypothesis_def}-\eqref{eq:psi_limits}. These approximations can then be iteratively refined, ultimately converging to an optimal solution. 
While our approach in this work is refined specifically for $({\boldsymbol{\mathcal{P}}})$, it can be applied for general link formation problems in decentralized federated domain adaptation, where formulations are concerned with optimizing ML performance through network control.  
We are among the first to leverage these flexible optimization methods for domain adaptation problems. 

Our methodology exploits approximations of~\eqref{eq:obj_fxn_1}-\eqref{eq:psi_limits} to convert $({\boldsymbol{\mathcal{P}}})$ from a mixed-integer signomial programming problem to a geometric programming problem, which, after a logarithmic change of variables, becomes a convex programming problem that can be solved using modern optimization libraries such as CVXPY~\cite{diamond2016cvxpy}. 
Thus, to properly explain our methodology, we must discuss geometric programming (GP), which requires first defining monomials and posynomials. 
\begin{definition} \label{defn:monomial_posynomial}
A \textbf{monomial} is defined as a function\footnote{$\mathbb{R}^n_{++}$ denotes the strictly positive quadrant of an $n$-dimensional Euclidean space.} $f: \mathbb{R}^n_{++}\rightarrow \mathbb{R}$ of the form $f(\bm{y})=z y_1^{\beta_1} y_2^{\beta_2} \cdots y_n ^{\beta_n}$, where $z\geq 0$, $\bm{y}=[y_1,\cdots,y_n]$, and $\beta_j\in \mathbb{R}$, $\forall j$. 
A \textbf{posynomial} $g$ is defined as a sum of monomials, and has form $g(\bm{y})= \sum_{m=1}^{M} z_m y_1^{\beta^m_1} y_2^{\beta^m_2} \cdots y_n ^{\beta^m_n}$, $z_m \geq 0,\forall m$.
\end{definition}
\noindent
A further discussion of GP is available in Appendix~\ref{app_sssec:gp}, but the key point is that only standard GP with posynomial objective function subject to posynomial inequality and monomial equality constraints can enable a logarithmic change of variables and subsequent solution using modern solvers. Our formulation, $({\boldsymbol{\mathcal{P}}})$, violates GP rules. Specifically, the objective function~\eqref{eq:obj_fxn_1} is not a posynomial or a monomial due to the negative optimization variables (i.e., $(1-\psi_i)$) present in parts $(c)$ and $(d)$. Furthermore, both~\eqref{eq:only_targets_receive} and~\eqref{eq:def_E_i} violate GP constraints, as the former is an equality constraint on posynomials while the latter contains a division by a posynomial which is not a posynomial. 
To address all these violating terms, we use the method of penalty functions and auxiliary optimization variables~\cite{xu2014global} to approximate (i) the negative optimization variables, (ii) the posynomial divisors, and (iii) the posynomial equality constraints.

\begin{table*}[tbp]
\begin{minipage}{0.99\textwidth}
{\scriptsize
\begin{equation}\label{eq:approx_source_error}
\begin{aligned}
    F_i(\bm{x}) = \psi_i + \frac{\chi^{\mathsf{S}}_i}{S_i} \geq \widehat{F}_i(\bm{x};\ell) \triangleq \left(\frac{\psi_i F_i([\bm{x}]^{\ell-1}) } {[\psi_i]^{\ell-1}} \right) ^ {\frac{[\psi_i]^{\ell-1}}{F_i([\bm{x}]^{\ell-1})}} 
    \left( \frac{\chi^{\mathsf{S}}_i F_i([\bm{x}]^{\ell-1})}{[\chi^{\mathsf{S}}_i]^{\ell-1}} \right)^ {\frac{[\chi^{\mathsf{S}}_i/ S_i]^{\ell-1}}{F_i([\bm{x}]^{\ell-1})}}
\end{aligned}
\vspace{-1mm}
\end{equation}
\vspace{-1mm}
\hrulefill
\begin{equation}\label{eq:approx_target_error1}
\begin{aligned}
    G_{i,j,k}(\bm{x}) = \psi_k + \frac{\widehat{\chi}^{\mathsf{T}}_{i,j,k}}{\psi_i \alpha_{k,i}\widehat{T}_{i,j,k}} \geq 
    \widehat{G}_{i,j,k}(\bm{x};\ell) 
    \triangleq \left( \frac{\psi_k G_{i,j,k}([\bm{x}]^{\ell-1})}{[\psi_k]^{\ell-1}} \right) ^ {\frac{[\psi_k]^{\ell-1}}{G_{i,j,k}([\bm{x}]^{\ell-1})}}
    \left( \frac{ G_{i,j,k}([\bm{x}]^{\ell-1}) \widehat{\chi}^{\mathsf{T}}_{i,j,k} /(\psi_i \alpha_{k,i}) } 
    {\left[ \widehat{\chi}^{\mathsf{T}}_{i,j,k}/ (\psi_i \alpha_{k,i}) \right]^{\ell-1}} \right) ^ 
    {\frac{\left[ \widehat{\chi}^{\mathsf{T}}_{i,j,k}/ (\widehat{T}_{i,j,k}\psi_i \alpha_{k,i}) \right]^{\ell-1}} {G_{i,j,k}([\bm{x}]^{\ell-1})}} \hspace{-2mm}
\end{aligned}
\vspace{-1mm}
\end{equation}
\vspace{-1mm}
\hrulefill
\begin{equation}\label{eq:approx_target_error2}
\begin{aligned}
    & H_{i,j}(\bm{x}) = \psi_i \widehat{T}_{i,j} + \psi_i \sum_{k \in \mathcal{N}} \widehat{\chi}^{\mathsf{T}}_{i,j,k} + \frac{\chi^{\mathsf{T}}_{i,j}}{\psi_j\alpha_{i,j}} \geq 
    \widehat{H}_{i,j}(\bm{x};\ell) \triangleq \\
    & \left( \frac{\psi_i H_{i,j}([\bm{x}]^{\ell-1})}{[\psi_i]^{\ell-1}}  \right) ^ { \frac{[\psi_i]^{\ell-1} \widehat{T}_{i,j}} {H_{i,j}([\bm{x}]^{\ell-1})}} 
    \left( \frac{ H_{i,j}([\bm{x}]^{\ell-1}) \chi^{\mathsf{T}}_{i,j}/ (\psi_j\alpha_{i,j})  } {\left[ \chi^{\mathsf{T}}_{i,j} / (\psi_j\alpha_{i,j}) \right]^{\ell-1}} \right) ^ {\frac{\left[ \chi^{\mathsf{T}}_{i,j} / (\psi_j\alpha_{i,j})\right]^{\ell-1}} {H_{i,j}([\bm{x}]^{\ell-1})} } 
    \prod_{k \in \mathcal{N}} \left( \frac{\psi_i \widehat{\chi}^{\mathsf{T}}_{i,j,k} H_{i,j}([\bm{x}]^{\ell-1})} 
    {[\psi_i \widehat{\chi}^{\mathsf{T}}_{i,j,k}]^{\ell-1} } \right) ^{ \frac{[\psi_i \widehat{\chi}^{\mathsf{T}}_{i,j,k}]^{\ell-1}}{H_{i,j}([\bm{x}]^{\ell-1})}}
\end{aligned}
\vspace{-1mm}
\end{equation}
\vspace{-1mm}
\hrulefill
\begin{equation}\label{eq:approx_nrg}
\begin{aligned}
    J_{i,j}(\bm{x}) = \alpha_{i,j} + \epsilon_{E} 
    \geq \widehat{J}_{i,j}(\bm{x};\ell) 
    \triangleq \left( \frac{\alpha_{i,j} J_{i,j}([\bm{x}]^{\ell-1})}{[\alpha_{i,j}]^{\ell-1}} \right) ^ {\frac{[\alpha_{i,j}]^{\ell-1}}{J_{i,j}([\bm{x}]^{\ell-1})}} 
    \left( J_{i,j}([\bm{x}]^{\ell-1}) \right) ^ {\frac{\epsilon_E}{J_{i,j}([\bm{x}]^{\ell-1})}},
\end{aligned}
\vspace{-1mm}
\end{equation}
\vspace{-1mm}
\hrulefill
\begin{equation}\label{eq:approx_con_plus}
\begin{aligned}
    & M^{+}_{i,j}(\bm{x}) = \chi^{C}_{i,j} + \epsilon_{C} + \psi_j \geq 
    \widehat{ M}^{+}_{i,j}(\bm{x};\ell) \triangleq 
    \left( \frac{\chi^{C}_{i,j} M^{+}_{i,j}([\bm{x}]^{\ell-1})}{[\chi^{C}_{i,j}]^{\ell-1}} \right) ^{\frac{[\chi^{C}_{i,j}]^{\ell-1}}{M^{+}_{i,j}([\bm{x}]^{\ell-1})}} 
    \left( M^{+}_{i,j}([\bm{x}]^{\ell-1}) \right) ^{\frac{\epsilon_{C}}{M^{+}_{i,j}([\bm{x}]^{\ell-1})}}
    \left( \frac{\psi_j M^{+}_{i,j}([\bm{x}]^{\ell-1})}{[\psi_j]^{\ell-1}} \right) ^{\frac{[\psi_j]^{\ell-1}}{M^{+}_{i,j}([\bm{x}]^{\ell-1})}} \hspace{-2mm}   
\end{aligned}
\vspace{-1mm}
\end{equation}
\vspace{-1mm}
\hrulefill
\begin{equation}\label{eq:approx_con_minus}
\begin{aligned}
    & M^{-}_{i,j}(\bm{x}) = \sum_{i \in \mathcal{N}} \alpha_{i,j} \geq 
    \widehat{ M}^{-}_{i,j}(\bm{x};\ell) \triangleq 
    \prod_{i \in \mathcal{N}} \left( \frac{\alpha_{i,j} M^{-}_{i,j}([\bm{x}]^{\ell-1})}{[\alpha_{i,j}]^{\ell-1}} \right) ^{\frac{[\alpha_{i,j}]^{\ell-1}}{M^{-}_{i,j}([\bm{x}]^{\ell-1})}}  \hspace{-2mm}   
\end{aligned}
\vspace{-1mm}
\end{equation}
\vspace{-1mm}
\hrulefill
}
\end{minipage}
\vspace{-3mm}
\end{table*}

We first consider terms with negative optimization variables. First, we bound these terms using a unique auxiliary variable, yielding the general format: 
$a(x) - b(x) < \chi$, where $a(x)$ represents terms with positive variables, $b(x)$ represents terms with negative variables, and $\chi > 0$ is the auxiliary variable. 
Manipulating this expression yields $a(x)/(b(x) + \chi) \leq 1$, which is not a posynomial due to the division by a posynomial. 
One way to proceed is to approximate the posynomial denominator using a monomial, which allows the expression to become a posynomial (as the division of a posynomial by a monomial is a posynomial). To do so, we can bound a posynomial with a monomial through the arithmetic-geometric mean inequality: 
\begin{lemma}[\textbf{Arithmetic-geometric mean inequality~\cite{duffin1972reversed}}]
\label{Lemma:AG_mean}
Consider a posynomial function $g(\bm{y})=\sum_{i=1}^{i'} u_i(\bm{y})$, where $u_i(\bm{y})$ is a monomial, $\forall i$. The following inequality holds:
\begin{equation}\label{eq:approxPosMonMain}
    g(\bm{y})\geq \widehat{g}(\bm{y})\triangleq \prod_{i=1}^{i'} \left( \frac{u_i(\bm{y})}{\alpha_i(\bm{z})} \right)^{\alpha_i(\bm{z})},
\end{equation}
where $\alpha_i(\bm{z})=u_i(\bm{z})/g(\bm{z})$, $\forall i$, and $\bm{z}>0$ is a fixed point.
\end{lemma}

\noindent 
The arithmetic-geometric mean inequality yields an initial, monomial bound that can be rather loose, but can be iteratively refined to obtain near equality with the original posynomial. 
We repeat this posynomial-to-monomial approximation process for terms $(c)$ and $(d)$ from~\eqref{eq:obj_fxn_1}, as they contain negative optimization variables. Additionally, $T_{i,j}$ contains $h^{\mathsf{T}}_j$, which contains another negative $(1-\psi_j)$ term, so we use an approximation for $T_{i,j}$ within our approximation for $(d)$ from~\eqref{eq:obj_fxn_1}. 
As~\eqref{eq:def_E_i} also has a posynomial denominator, we apply the same strategy to approximate its posynominal denominator as a monomial. 
Finally, we can substitute the equality constraint in~\eqref{eq:only_targets_receive} into two constraints (i) $\sum_{i \in \mathcal{N}} \alpha_{i,j} - \psi_j \leq \chi^{C}_{i,j} + \epsilon_{C}$ and (ii) $\sum_{i \in \mathcal{N}} \alpha_{i,j} - \psi_j \geq \chi^{C}_{i,j} - \epsilon_{C}$, with auxiliary variable $\chi^{\mathsf{C}}_{i,j}$ and a very small constant $\epsilon_C > 0$. In this way, the objective function augmented with $\chi^{\mathsf{C}}_{i,j}$ will squeeze (i) and (ii), approximating equality. Both constraints (i) and (ii) involve negative optimization variables, allowing us to then approximate them using the above outlined process for~\eqref{eq:obj_fxn_1}. 

To summarize, the approximations for~\eqref{eq:obj_fxn_1}-\eqref{eq:def_E_i} are displayed in~\eqref{eq:approx_source_error}-\eqref{eq:approx_con_minus}. 
The above discussion has been a high-level overview of the derivations of our approximations. The full derivations for each approximation are rather lengthy, and are thus deferred to Appendix~\ref{app_sssec:tx_2gp} for conciseness. 
To ensure that these approximations converge independently, they are associated with a unique auxiliary variable that is added to the objective function or bounded between very small quantities using an additional pair of constraints. 
The resulting optimization formulation with modified objective function, approximations~\eqref{eq:approx_source_error}-\eqref{eq:approx_con_minus}, augmented constraints, and auxiliary variables retains the same core insights as those for $({\boldsymbol{\mathcal{P}}})$, and is thus similarly left to Appendix~\ref{app_sssec:tx_2gp} to minimize repetition. 

As a result of the approximations in~\eqref{eq:approx_source_error}-\eqref{eq:approx_con_minus}, our optimization solver, summarized in Algorithm~\ref{alg:optimization_iteration}, is an iterative method that starts with an initial value for the solution $[\bm{x}]^0$, containing an initial estimate for $\boldsymbol{\alpha}$ and $\boldsymbol{\psi}$. The only requirement for the initial estimate is that it is feasible for $\boldsymbol{\alpha}$ and $\boldsymbol{\psi}$ - our solver will then converge to the optimal regardless of initial estimate. 
At each iteration indexed via $\ell$, our solver obtains a solution $[\bm{x}]^\ell$ via transforming the problem into a solvable convex program, in which all the terms in the objective function and constraints are transformed into convex terms around the previous $[\bm{x}]^{\ell-1}$. }

\vspace{-1mm}
\section{Numerical Evaluation}
\label{sec:experiments}

\noindent 
{\color{black} In this section, we experimentally demonstrate four key points of ST-LF.
In Sec.~\ref{ssec:ablate_hat_P}, we show that our optimization $({\boldsymbol{\mathcal{P}}})$ 
(i) allocates source/target classification $\boldsymbol{\psi}$ and link formation $\boldsymbol{\alpha}$ effectively based on distribution divergence $\widehat{d}_{\mathcal{H}\Delta\mathcal{H}}$, and 
(ii) adjusts communication resource consumption as a result of both the energy scaling $\phi^{\mathsf{E}}$ and the underlying classification task. 
Then, in Sec.~\ref{ssec:ablate_tx_learn}, we show that our methodology (iii) obtains significant enhancements in classification accuracy at target devices for domain adaptation tasks and (iv) significant communication resource savings compared to several baseline algorithms, validating its benefit for decentralized FL settings.



\begin{figure}[t]
    \centering
    \includegraphics[width=0.48\textwidth]{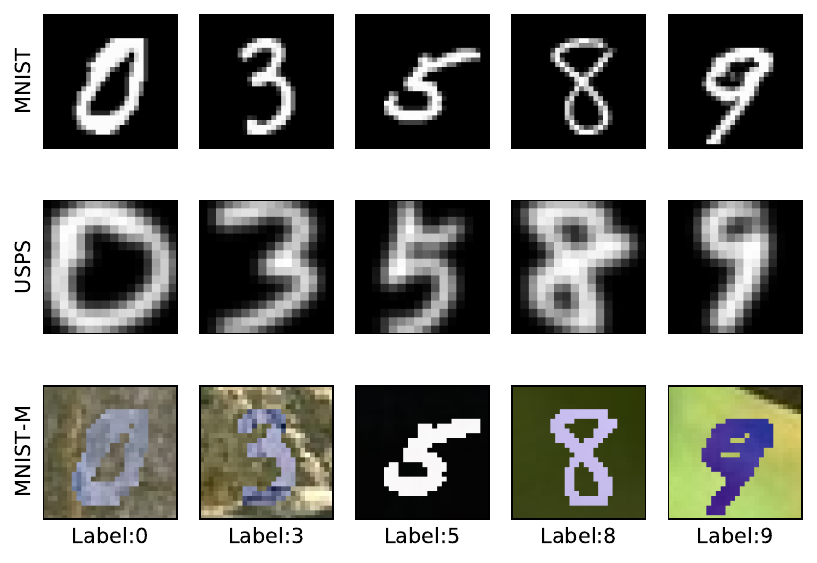} 
    \vspace{-2mm}
    \caption{An overview of the three common domain adaptation datasets used to evaluate our method. We explain the physical differences of each dataset in Sec.~\ref{sec:experiments}.}
    \label{fig:data_snapshot}
\vspace{-3mm}
\end{figure}

\textbf{Experimental Setup.} 
We conduct our evaluations on three image classification datasets commonly used in domain adaptation: MNIST~\cite{lecun1998gradient}, USPS~\cite{hull1994database}, and MNIST-M~\cite{ganin2016domain}.
A snapshot of the three datasets is shown in~\ref{fig:data_snapshot}.
In particular, while all three datasets are concerned with the problem of digit recognition, the images in each dataset are formatted differently, and, thus, transferring a trained ML model from one dataset to another encounters some difficulty. 
Specifically, the images in MNIST are neatly formatted with the same background, and USPS contains images similar to MNIST but taken with different resolution. Finally, as a tougher challenge, MNIST-M contains digits from MNIST blended with various photographs from BSDS500~\cite{arbelaez2010contour}, resulting in a wide-range of digit outlines and backgrounds (e.g., the sample MNIST-M image with label 3 in Fig.~\ref{fig:data_snapshot} is challenging even for people to decipher). 

\begin{figure}[t]
\centering
\begin{subfigure}{.24\textwidth} 
  \centering
  \includegraphics[width=.96\linewidth]{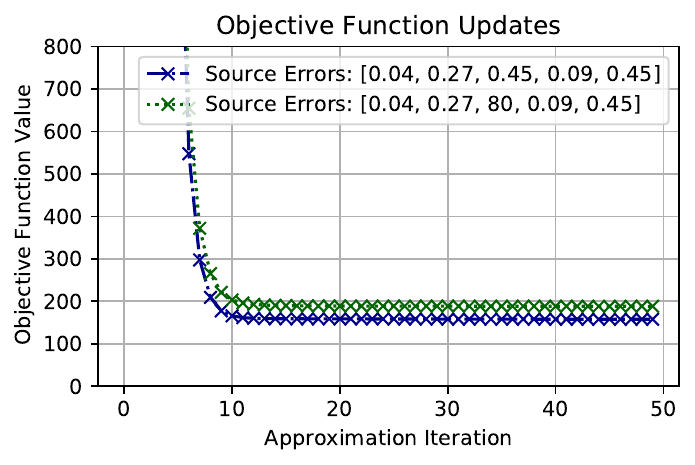}
  \vspace{-1mm}
  \caption{Convergence of $(\boldsymbol{\mathcal{P}})$.} 
  \label{fig:obj_fxn_convergence}
\end{subfigure}
\begin{subfigure}{.24\textwidth}
  \centering
  \includegraphics[width=.96\linewidth]{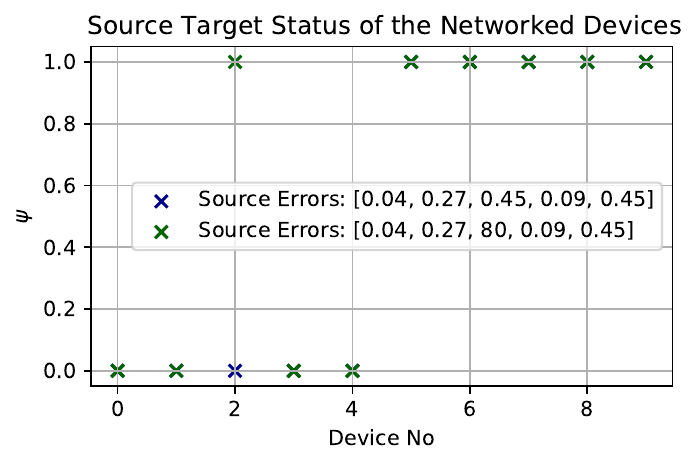}
    \vspace{-1mm}
  \caption{Source/target classification $\psi_i$.}
  \label{fig:st_determination_source_error}
\end{subfigure}
\vspace{-1.5mm}
\caption{{\color{black}Convergence behavior and source/target device classification at convergence for Algorithm~\ref{alg:optimization_iteration} with two different settings of source errors across devices with labeled/partially labeled data.}}
\vspace{-3mm}
\end{figure}

\begin{figure}[t]
    \centering
    \includegraphics[width=0.498\textwidth]{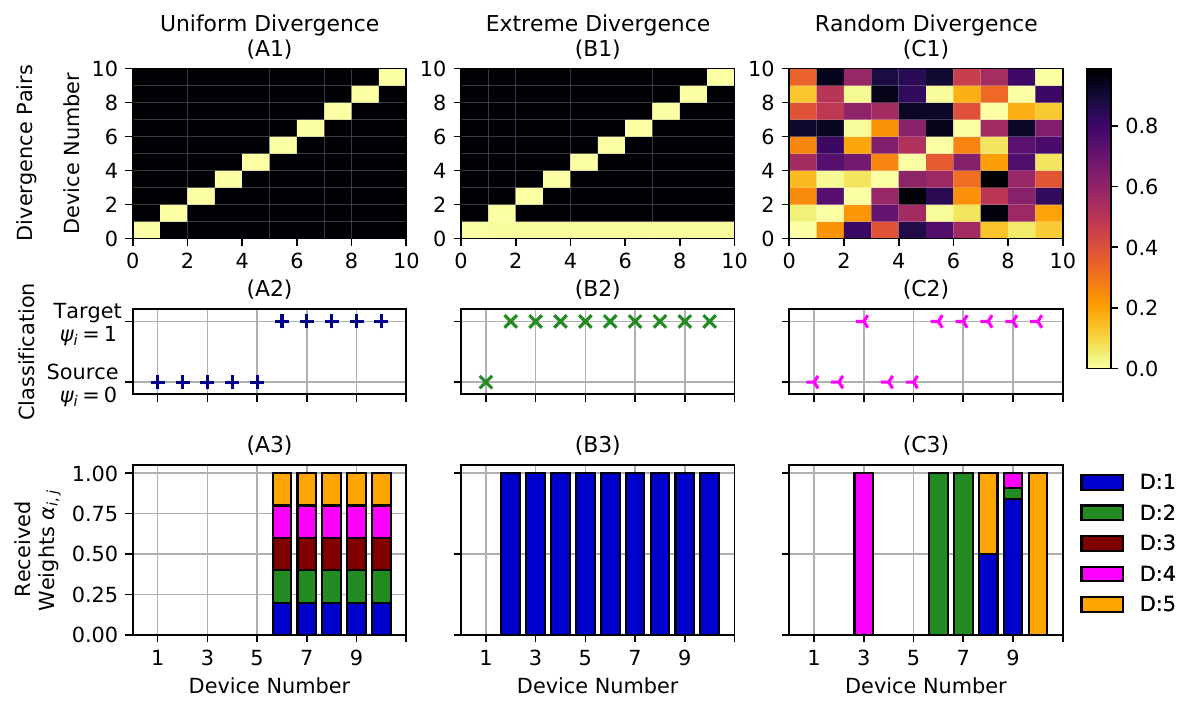}
    \vspace{-2mm}
    \caption{{\color{black}The effects of uniform, extreme, and random distribution divergence regimes on the behavior of $({\boldsymbol{\mathcal{P}}})$. 
    Each regime occupies a column, showing (i) the divergences $d_{\mathcal{H}\Delta\mathcal{H}}$ between pairs of devices, (ii) the optimized source/target classifications $\boldsymbol{\psi}$, and (iii) the optimized combination weights $\boldsymbol{\alpha}$. 
    The third row breaks down received ML models at targets from source devices 1-5 (i.e., $D:1,\cdots,5$), which are proportional to the divergences.}}
    \label{fig:3x3_div}
\vspace{-4mm}
\end{figure}

We consider a network of 10 devices, training two-layer CNNs (with 10 and 20 maps respectively) followed by two fully connected layers for all cases. For our federated divergence estimation (Algorithm~\ref{alg:fed_div_est}), we use a CNN with the same architecture as above, except that the dimension at the output of the final connected layer is $2$ rather than $10$. 
Next, we distribute data across the devices in a non-i.i.d. manner, where each device has a unique Dirichlet distribution of all labels or subset of labels from the full training dataset~\cite{wang2021novel}. Half of the network will have partially labeled datasets with randomly determined labeled-to-unlabeled data ratios, while the rest of the network will have completely unlabeled datasets. 
For the singular dataset tests involving MNIST and USPS, devices draw Dirichlet distributions from a subset of $4$ labels, while for simulations involving MNIST-M and/or multiple simultaneous datasets, devices draw from the full training dataset. 
We locally train ML models at source devices using stochastic gradient descent (SGD) as in conventional FL, with $100$ iterations, a mini-batch size of $10$, and a learning rate of $0.01$. Training is conducted with Pytorch~\cite{paszke2019pytorch} on a 6GB GTX 1660 SUPER with 16 GB RAM. 
Additionally, for ST-LF, we use $\phi^{\mathsf{S}}=1$, $\phi^{\mathsf{T}}=5$, and $\phi^{\mathsf{E}}=1$ in $({\boldsymbol{\mathcal{P}}})$, and all figures are the averaged results over 5 independent runs, unless otherwise stated. 

\begin{figure*}[t]
\centering
\begin{subfigure}{.32\textwidth}
  \centering
  \includegraphics[width=.96\linewidth]{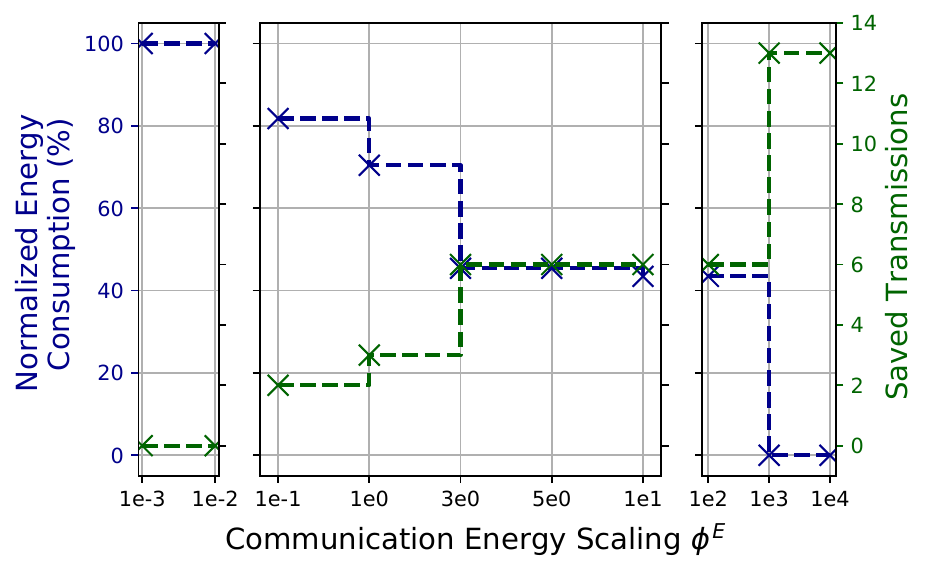}
    \vspace{-1mm}
  \caption{MNIST}
  \label{fig:nrg_vary_phi_e_M}
\end{subfigure}
\begin{subfigure}{.32\textwidth}
  \centering
  \includegraphics[width=.96\linewidth]{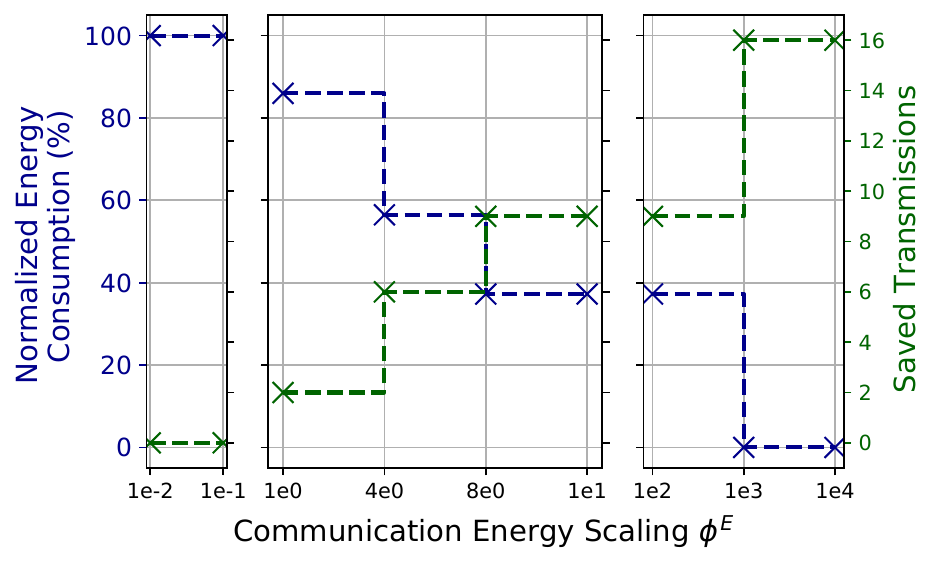}
    \vspace{-1mm}
  \caption{USPS}
  \label{fig:nrg_vary_phi_e_U}
\end{subfigure}
\begin{subfigure}{.32\textwidth}
  \centering
  \includegraphics[width=.96\linewidth]{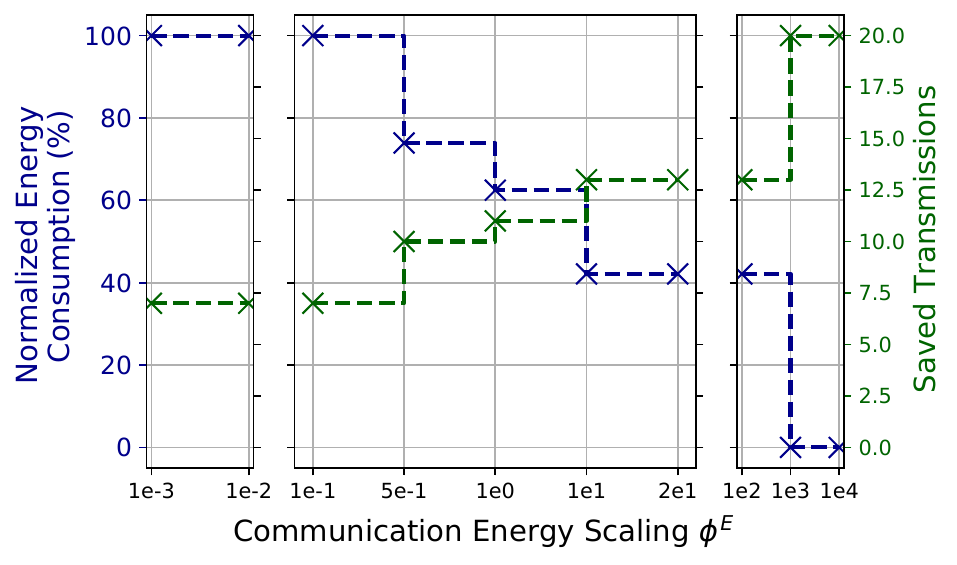}
    \vspace{-1mm}
  \caption{MNIST-M}
  \label{fig:nrg_vary_phi_e_MM}
\end{subfigure}
\vspace{-1.5mm}
\caption{{\color{black}Total device communication resource consumption and reduction in transmissions as a function of the energy cost of ML model transmissions ($\phi^{\mathsf{E}}$) in $({\boldsymbol{\mathcal{P}}})$. As $\phi^{\mathsf{E}}$ increases, ST-LF adjusts the combination weights $\boldsymbol{\alpha}$ so that fewer sources transfer their models to targets, with the energy consumption decreasing accordingly. Each dataset has different sensitivity to $\phi^{\mathsf{E}}$.}}
\label{fig:phi_e_vary}
\vspace{-3mm}
\end{figure*}

\begin{figure*}[t]
\centering
\begin{subfigure}{.32\textwidth}
  \centering
  \includegraphics[width=.76\linewidth]{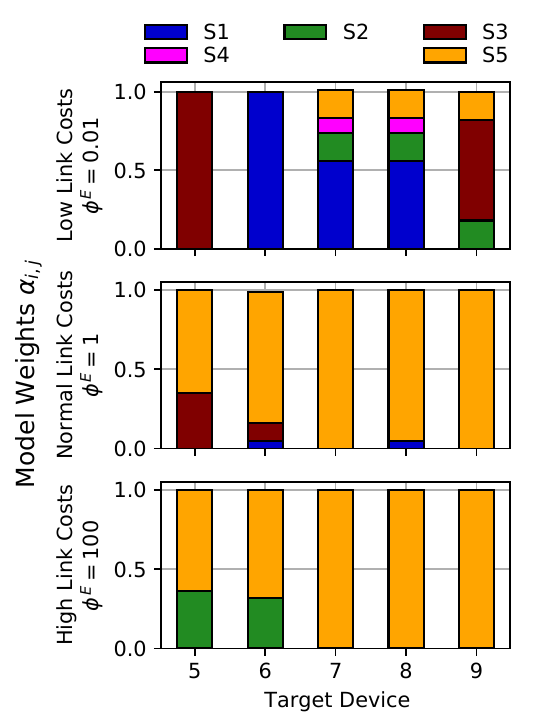}
    \vspace{-1mm}
  \caption{MNIST}
  \label{fig:nrg_alpha_M}
\end{subfigure}
\begin{subfigure}{.32\textwidth}
  \centering
  \includegraphics[width=.76\linewidth]{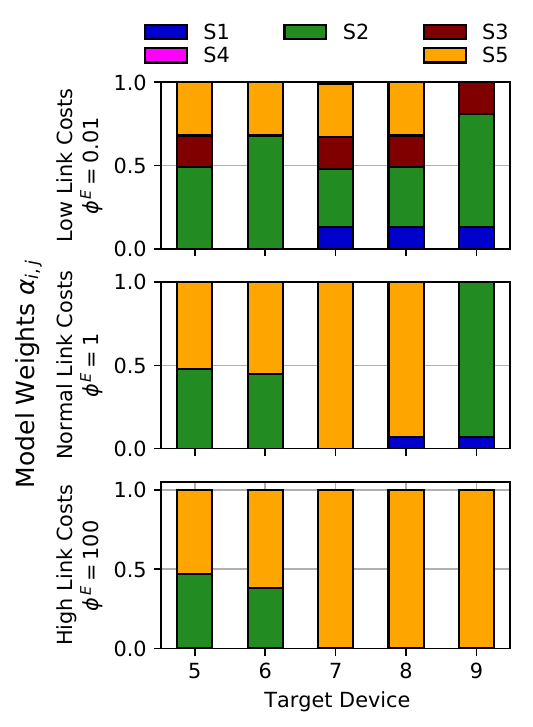}
    \vspace{-1mm}
  \caption{USPS}
  \label{fig:nrg_alpha_U}
\end{subfigure}
\begin{subfigure}{.32\textwidth}
  \centering
  \includegraphics[width=.76\linewidth]{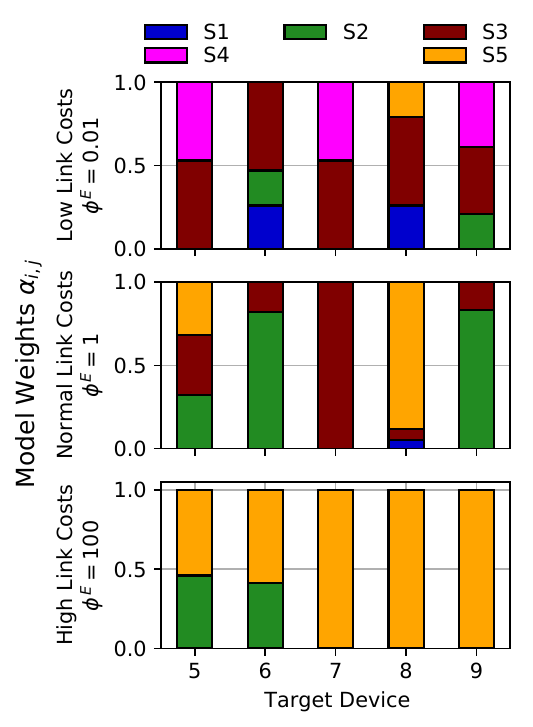}
    \vspace{-1mm}
  \caption{MNIST-M}
  \label{fig:nrg_alpha_MM}
\end{subfigure}
\vspace{-1.5mm}
\caption{
{\color{black}Effect of communication energy scaling, (i) low cost - $\phi^{\mathsf{E}} = 0.01$, (ii) medium cost - $\phi^{\mathsf{E}} = 1$, and (iii) high cost - $\phi^{\mathsf{E}}=100$, on the link formation weights $\boldsymbol{\alpha}$ for each of the datasets. The network and experiment setup are identical to those of Fig.~\ref{fig:phi_e_vary}.
As $\phi^{\mathsf{E}}$ increases, we see new combination weights involving fewer unique source-to-target transmissions in order to reduce communication costs.}}
\label{fig:nrg_ratios_phi_e}
\vspace{-3mm}
\end{figure*}

\textbf{Communication Energy Determination.}
The communication energy $E_{i,j}$ defined in~\eqref{eq:def_E_i} relies on $K_{i,j}$, a strictly communication-dependent term.
For any given device pair $(i,j)$, we define $K_{i,j}= \frac{M}{R_{i,j}}P_i$, where $P_i$ is transmission power at device $i$ (watts), $R_{i,j}$ is transmission rate from device $i$ to device $j$ (bits/sec), and $M$ is the size of the hypothesis (i.e., ML model) transmitted in bits. 
For all devices $i \in \mathcal{N}$, we determine the transmission power $P_i$ randomly between $P_{min} = 23$ \textrm{dBm} and $P_{max} = 25 $ {\textrm{dBm}}. 
Similarly, for all device pairs $(i,j)$, we select the transmission rate $R_{i,j}$ randomly between $R_{min} = 63$ Mbps and $R_{max} = 85$ Mbps. 
The size of the hypothesis in bits is assumed to be fixed at $1$ \textrm{Gbit}. 
In particular, these measurements are similar to commonly used measurements in autonomous systems research~\cite{mozaffari2017mobile}, and their selection highlights the versatility of our formulation and optimization methodologies. 

}

\subsection{ST-LF Optimization Solution} 
\label{ssec:ablate_hat_P}
\vspace{-.1mm}
We study multiple aspects of $({\boldsymbol{\mathcal{P}}})$ in-depth. We first demonstrate the convergence of our iterative methodology for $(\boldsymbol{\mathcal{P}})$ (Algorithm~\ref{alg:optimization_iteration}), and thereafter focus on our proposed federated divergence estimation technique (Algorithm~\ref{alg:fed_div_est}), which influences source-target classification and link formation weights. Finally, we examine the effect of $\phi^{\mathsf{E}}$ on the network communication resource consumption.

\textbf{Optimization Convergence and Source Error Sensitivity.} 
In Fig.~\ref{fig:obj_fxn_convergence} and Fig.~\ref{fig:st_determination_source_error}, we investigate (i) the convergence of our iterative methodology for $(\boldsymbol{\mathcal{P}})$ and (ii) the effect of having a large empirical error when a device has labelled data. For both figures, we assume that five devices have labeled or partially labeled datasets and the other five devices only have unlabeled data. We first show that, regardless of empirical error on a labeled dataset, our solution for $({\boldsymbol{\mathcal{P}}})$ is able to reach convergence in a monotonically decreasing fashion. Then, in Fig.~\ref{fig:st_determination_source_error}, we show that a large empirical error at device 3 results in $({\boldsymbol{\mathcal{P}}})$ classifying that device as a target even though it may have labeled data. Essentially, our optimization is able to balance the trade-offs between having more sources (i.e., more training data) versus the quality of those sources. In scenarios where a device has a small quantity and/or poor quality of locally labeled data, it may be more effective from a network resource perspective to classify that device as a target.

\textbf{Adapting to Distribution Divergence.} 
Fig.~\ref{fig:3x3_div} demonstrates the effect of distribution divergence $\widehat{d}_{\mathcal{H}\Delta\mathcal{H}}$ on ST-LF's optimization solution. Recall from Sec.~\ref{ssec:init_formulation} and Algorithm~\ref{alg:fed_div_est} that a large divergence means that two devices are statistically dissimilar while a small divergence implies that they have highly similar local dataset characteristics. 
We depict three separate regimes in the columns of Fig.~\ref{fig:3x3_div}: (A) \textit{uniform:} devices have identical pairwise divergences of $1$; (B) \textit{extreme:} one device has the minimum pairwise divergence, $0$, to all other devices while other pairs have the maximum, $1$; and (C) \textit{random:} device pairs have randomly assigned divergences. 
The first row of Fig.~\ref{fig:3x3_div} are 2D colormaps of the divergence values, where element $(i,j)$ indicates $\widehat{d}_{\mathcal{H}\Delta \mathcal{H}}(i,j)$, while the second and third row depict the resulting source/target classifications $\boldsymbol{\psi}$ and combination weights $\boldsymbol{\alpha}$ from solving $({\boldsymbol{\mathcal{P}}})$.

\begin{figure*}[t]
\centering
\begin{subfigure}{.325\textwidth}
  \centering
  \includegraphics[width=.98\linewidth]{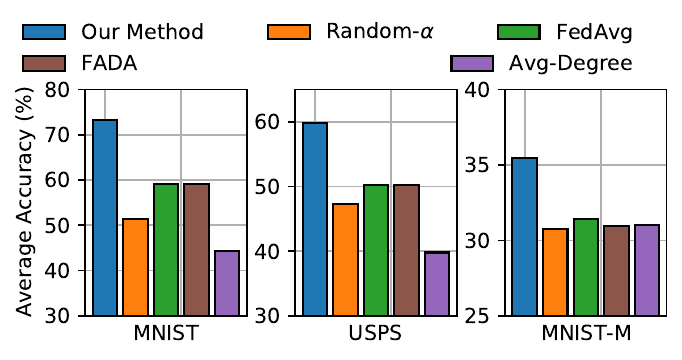}
  \vspace{-1mm}
  \caption{Single dataset}
  \label{fig:avg_gan}
\end{subfigure}
\begin{subfigure}{.325\textwidth}
  \centering
  \includegraphics[width=.98\linewidth]{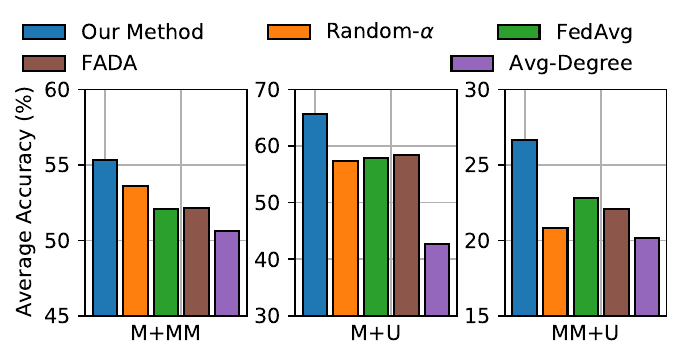}
  \vspace{-1mm}
  \caption{Mixed dataset}
  \label{fig:mixed_gan}
\end{subfigure}
\begin{subfigure}{.325\textwidth}
  \centering
  \includegraphics[width=.98\linewidth]{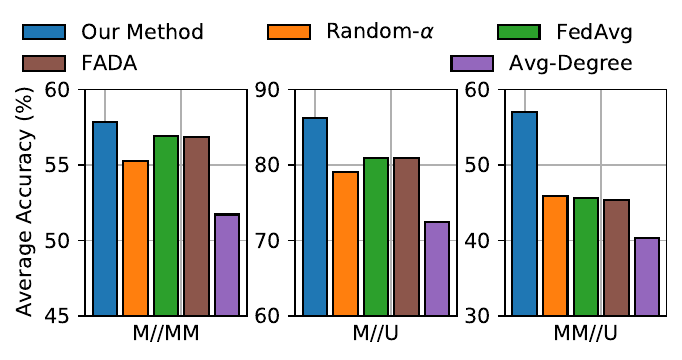}
  \vspace{-1mm}
  \caption{Split dataset}
  \label{fig:split_gan}
\end{subfigure}
\vspace{-1mm}
\caption{
{\color{black}
Comparing the model weight determination, $\boldsymbol{\alpha}$, of our method (ST-LF) versus several baselines (e.g., FedAvg~\cite{mcmahan2017communication}, FADA~\cite{peng2019federated}) on single, mixed, and split dataset settings. All methods rely on ST-LF's source/target classification output ($\boldsymbol{\psi}$), with each bar representing the average classification accuracy across all target devices obtained over five independent runs of each algorithm. 
Mixed datasets are denoted by ``+'', e.g., M+U means devices have a mixture of data from the MNIST and USPS datasets, while split datasets are described by ``//'', e.g., M//U means that some devices only have data from the MNIST dataset while others will exclusively contain data from the USPS dataset.}}
\label{fig:mt_alpha_all}
\vspace{-3mm}
\end{figure*}

\begin{figure*}[t]
\centering
\begin{subfigure}{.325\textwidth}
  \centering
  \includegraphics[width=.98\linewidth]{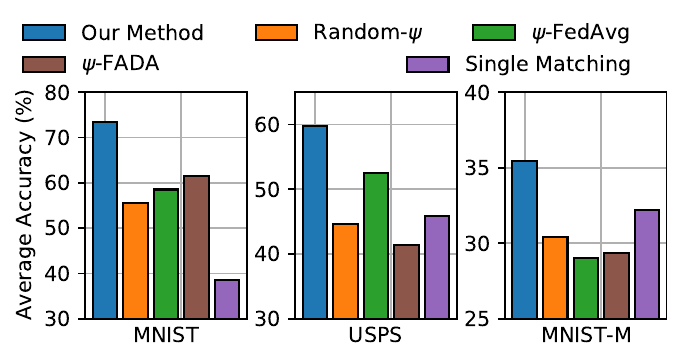}
  \vspace{-1mm}
  \caption{Single dataset}
  \label{fig:avg_gan_full}
\end{subfigure}
\begin{subfigure}{.325\textwidth}
  \centering
  \includegraphics[width=.98\linewidth]{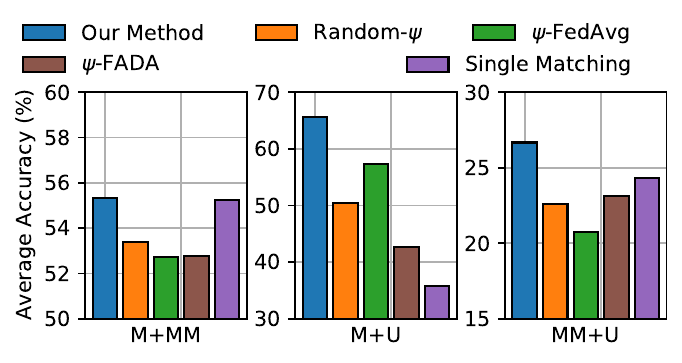}
  \vspace{-1mm}
  \caption{Mixed dataset}
  \label{fig:mixed_gan_full}
\end{subfigure}
\begin{subfigure}{.325\textwidth}
  \centering
  \includegraphics[width=.98\linewidth]{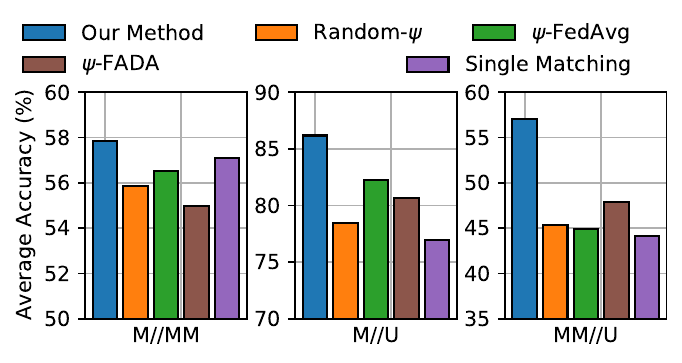}
  \vspace{-1mm}
  \caption{Split dataset}
  \label{fig:split_gan_full}
\end{subfigure}
\vspace{-1mm}
\caption{
{\color{black}
The impact of source/target classification, $\boldsymbol{\psi}$, on average target classification accuracy. When baseline methods (e.g., $\psi$-FL~\cite{mcmahan2017communication}, $\psi$-FADA~\cite{peng2019federated}) do not have access to optimal source/target determination from our method (ST-LF), they generally perform worse than in Fig.~\ref{fig:mt_alpha_all} relative to our method, demonstrating the importance of joint source/target classification, $\boldsymbol{\psi}$, and model/link weight determination, $\boldsymbol{\alpha}$. 
} }
\label{fig:mt_psi_all}
\vspace{-5mm}
\end{figure*}

In the uniform regime (A1-A3), devices are essentially statistically identical to each other, so targets should receive uniformly weighted model parameters from the sources. 
ST-LF adapts to this by classifying all devices with labeled data (i.e., devices $1$ to $5$) as sources and unlabeled devices as targets in (A2), and forming uniformly weighted links between all source-target pairs in (A3). 
In the extreme divergence regime (B1-B3), there is one device that is statistically perfect for every other device in the network. ST-LF detects this and correspondingly assigns device 1 as the only source, with all source-target combination weights set to $1$. 
With the random divergence regime (C1-C3), our methodology needs to be able to pick-out trade-offs among the devices. For example, if a device has labeled data but is very similar to the data at another device with labeled data, then it may be better to only keep one of these two devices as a source. 
ST-LF adapts well in this more challenging setting, determining that device 3 has high statistical similarity with device 4, and thus reclassifies device 3 as a target in (C2). 
Finally, the link weights $\boldsymbol{\alpha}$, which also reflect model combination weights, vary based on the $d_{\mathcal{H}\Delta\mathcal{H}}$ values, which makes sense as target devices should favor the trained models at statistically similar sources devices.

\textbf{Impact of Communication Energy Importance.} 
Communication energy expenditures in ST-LF follow a discrete pattern, where the network only sees a change in both total transmission and energy consumption at distinct thresholds. 
For example, when $\alpha_{i,j}$ changes numerically from $0.2$ to $0.6$, the link between device $i$ and device $j$ incurs the same energy cost, though the weight of those ML model parameters changes.
To effect a change, $\alpha_{i,j}$ must become deactivated, i.e., goes to zero, or become activated, i.e., increases from zero to a non-zero value.
We investigate the ability of our methodology to control these thresholds via varying the communication energy scaling $\phi^{\mathsf{E}}$ from $({\boldsymbol{\mathcal{P}}})$ in Fig.~\ref{fig:phi_e_vary}.

Specifically, Fig.~\ref{fig:phi_e_vary} depicts the ability of $\phi^{\mathsf{E}}$ to influence total network transmissions and energy consumption. 
On the left y-axis, Fig.~\ref{fig:phi_e_vary} measures the normalized energy consumption, a metric computed from term $(e)$ in~\eqref{eq:obj_fxn_1}, while the right y-axis measures saved transmissions. 
Both normalized energy consumption and saved transmissions are measured relative to the case of $\phi^{\mathsf{E}}=0$. 
For the lowest values of $\phi^{\mathsf{E}}$ in Fig.~\ref{fig:phi_e_vary}, the energy consumption is saturated, effectively equivalent to solving ($\boldsymbol{\mathcal{P}}$) without the energy term $(e)$.
{\color{black}
As $\phi^{\mathsf{E}}$ increases, energy consumption decreases at discrete intervals, corresponding to exact points of $\phi^{\mathsf{E}}$ when links deactivate. 
Additionally, different datasets have inherently unique response patterns to $\phi^{\mathsf{E}}$: ST-LF saves up to 13, 16, and 20 transmissions for MNIST, USPS, and MNIST-M when $\phi^{\mathsf{E}} \geq 1e3$. 
Finally, we notice a saturation effect as $\phi^{\mathsf{E}}$ reaches $1e3$, as further increases in $\phi^{\mathsf{E}}$ do not reduce energy consumption or save any additional transmissions. This is because, when $\phi^{\mathsf{E}}$ exceeds $1e3$, links become prohibitively expensive to activate, so ST-LF designates all the devices as targets instead.} 

We visualize the change in weights $\alpha_{i,j}$ as a result of different $\phi^{\mathsf{E}}$ values in Fig.~\ref{fig:nrg_ratios_phi_e}. For example, for the MNIST experiment in Fig.~\ref{fig:nrg_alpha_M}, increasing $\phi^{\mathsf{E}}$ from $0.01$ to $1$ eliminates transmissions from source device $4$ (i.e., $S4$), and this leads to some energy savings and saved transmissions in Fig.~\ref{fig:nrg_vary_phi_e_M}. 
{\color{black}
Fig.~\ref{fig:nrg_ratios_phi_e} also shows that for scenarios with large $\phi^{\mathsf{E}}$, the cost of model transmissions becomes the optimization bottleneck, resulting in only two sources that transmit to targets. Conversely, for scenarios with low $\phi^{\mathsf{E}}$, the optimization is able to focus more on the expected ML performance, resulting in more sources that transmit their ML models to targets.}
{\color{black}
The lowest value $\phi^{\mathsf{E}} = 0.01$ corresponds to the case of full energy consumption from Fig.~\ref{fig:phi_e_vary}.
}

\subsection{ST-LF Accuracy and Efficiency Enhancements} 
\label{ssec:ablate_tx_learn} 
{\color{black}
To evaluate the full ST-LF methodology, we now compare its effectiveness, as measured by average classification accuracies at target devices and total network energy consumption, relative to eight baselines on three commonly used domain adaptation datasets (MNIST~\cite{lecun1998gradient}, USPS~\cite{uspsdataset}, MNIST-M~\cite{ganin2015unsupervised,ganin2016domain}) depicted in Fig.~\ref{fig:data_snapshot}. We will use the abbreviations M for MNIST, U for USPS, and MM for MNIST-M in Fig.~\ref{fig:mt_alpha_all}, Fig.~\ref{fig:mt_psi_all}, and Table~\ref{tab:alg_nrg_comp_all}. 
{\color{black}We also provide an empirical investigation on the tightness/looseness of the theoretical results in Appendix~\ref{app_sec:tl_bounds}}
}

{
Current literature in the field commonly assumes that source/target devices are known apriori, and, to the best of our knowledge, works have yet to investigate source/target determination. 
Instead, current literature such as FADA~\cite{peng2019federated} has focused on determining model transmission or link weights, $\boldsymbol{\alpha}$, from source devices to a single target device. 
As a result, two of the baselines that we consider, (i) FedAvg~\cite{mcmahan2017communication} and (ii) FADA~\cite{peng2019federated}, require source/target device classification prior to operation. In this case, we feed these two algorithms the source/target classification output from ST-LF, i.e., the $\boldsymbol{\psi}$ from our method. 
Two other baselines, (iii) $\psi$-FedAvg and (iv) $\psi$-FADA, rely on heuristic determinations of source/target classification $\psi$ combined with the model weight determination, $\alpha$, from FedAvg~\cite{mcmahan2017communication} and FADA~\cite{peng2019federated}. 
In order to ensure a fair comparison between FADA~\cite{peng2019federated} and ST-LF, FADA's generators have the same architecture as ST-LF's binary hypothesis classifiers, except that FADA's generators have a larger output layer dimension to accommodate its subsequent operations. 
In this manner, we are able to separately measure the benefits of joint source/target classification and link weight determination from ST-LF. 
We also provide four additional heuristic baselines, (v) Random-$\alpha$, (vi) Avg Degree, (vii) Random-$\psi$, and (viii) Single Matching; we will detail each of these baselines in their relevant sections. 
}

\begin{table}[t]
\caption{{\color{black}Average target accuracy vs. communication energy (Nrg) consumption obtained by algorithms depicted in Fig.~\ref{fig:mt_alpha_all} and~\ref{fig:mt_psi_all}. 
Communication energy is normalized (norm) relative to the maximum energy use in each category, and we exclude energy consumption measurements of server-dependent algorithms. 
The improvements that our methodology (ST-LF) obtains on both metrics emphasizes the importance of joint optimization in decentralized federated settings.}}
{\footnotesize
\begin{tabularx}{0.49\textwidth}
{m{0.05em} m{4.6em} m{2.8em} m{2.8em}| m{4em} m{2.8em} m{2.8em}} 
\toprule[.2em]
& \multicolumn{3}{c}{\textbf{Comparing $\psi$}} & \multicolumn{3}{c}{\textbf{Comparing $\alpha$}} \\
\cmidrule(lr){2-4} \cmidrule{5-7}
& \multirow{3}{*}{\textbf{Method}} & {\textbf{Avg}} & {\textbf{Norm}} & \multirow{3}{*}{\textbf{Method}} & {\textbf{Avg}} & {\textbf{Norm}}\\ 
& & {\textbf{Acc}} & \textbf{Nrg} & & {\textbf{Acc}} & \textbf{Nrg}\\
& & {\textbf{(\%)}} &  \textbf{(\%)} & & {\textbf{(\%)}} & \textbf{(\%)}\\ 
\midrule
\multirow{4}{*}{\rotatebox[origin=c]{90}{{\parbox[c]{8.8mm}{\scriptsize \bf{MNIST}}}}}
& Ours & \textbf{71.32} & {21.10} & Ours & \textbf{71.32} & \textbf{21.10}\\ 
& Rnd-$\psi$ & 55.82 & 94.87 & Rnd-$\alpha$ & 54.54 & 99.15\\ 
& $\psi$-FedAvg & 59.64 & --- & FedAvg & 60.95 & --- \\ 
& $\psi$-FADA & 63.18 & 93.79 & FADA & 61.40 & 100 \\
& SM & 51.07 & \textbf{19.38} & AvgD & 58.63 & 55.15 \\
\midrule
\multirow{4}{*}{\rotatebox[origin=c]{90}{{\parbox[c]{7.8mm}{\scriptsize \bf{USPS}}}}} 
& Ours & \textbf{60.23} & {22.20} & Ours & \textbf{60.23} & \textbf{22.20}\\ 
& Rnd-$\psi$ & 47.97 & 94.87 & Rnd-$\alpha$ & 52.40 & 99.18\\ 
& $\psi$-FedAvg & 55.55 & --- & FedAvg & 54.00 & --- \\ 
& $\psi$-FADA & 46.60 & 87.21 & FADA & 54.46 & 100\\
& SM & 44.72 & \textbf{19.38} & AvgD & 53.35 & 57.83\\
\midrule
\multirow{4}{*}{\rotatebox[origin=c]{90}{{\parbox[c]{12.8mm}{\scriptsize \bf{MNIST-M}}}}} 
& Ours & \textbf{36.44} & {27.70} & Ours & \textbf{36.44} & \textbf{27.70}\\ 
& Rnd-$\psi$ & 30.11 & 94.87 & Rnd-$\alpha$ & 29.69 & 99.38\\ 
& $\psi$-FedAvg & 29.64 & --- & FedAvg & 30.74 & --- \\ 
& $\psi$-FADA & 28.79 & 93.67 & FADA & 30.20 & 100\\
& SM & 32.05 & \textbf{19.38} & AvgD & 30.60 & 69.60\\
\midrule[.2em]
\multirow{4}{*}{\rotatebox[origin=c]{90}{{\parbox[c]{9.6mm}{\scriptsize \bf{M+MM}}}}} 
& Ours & \textbf{56.43} & {24.11} & Ours & \textbf{56.43} & \textbf{24.11}\\ 
& Rnd-$\psi$ & 54.02 & 94.87 & Rnd-$\alpha$ & 53.28 & 98.41\\ 
& $\psi$-FedAvg & 53.15 & --- & FedAvg & 52.00 & --- \\ 
& $\psi$-FADA & 54.23 & 87.78 & FADA & 52.00 & 100\\
& SM & 54.93 & \textbf{19.38} & AvgD & 52.28 & 70.69\\
\midrule 
\multirow{4}{*}{\rotatebox[origin=c]{90}{{\parbox[c]{7mm}{\scriptsize \bf{M+U}}}}} 
& Ours & \textbf{68.01} & {20.56} & Ours & \textbf{68.01} & \textbf{20.56}\\ 
& Rnd-$\psi$ & 49.93 & 94.87 & Rnd-$\alpha$ & 58.60 & 100\\ 
& $\psi$-FedAvg & 55.85 & --- & FedAvg & 59.61 & --- \\ 
& $\psi$-FADA & 40.79 & 82.44 & FADA & 59.90 & 100\\
& SM & 51.73 & \textbf{19.38} & AvgD & 56.49 & 55.15\\
\midrule
\multirow{4}{*}{\rotatebox[origin=c]{90}{{\parbox[c]{9mm}{\scriptsize \bf{MM+U}}}}} 
& Ours & \textbf{26.64} & {22.44} & Ours & \textbf{26.64} & \textbf{22.44}\\ 
& Rnd-$\psi$ & 22.04 & 94.42 & Rnd-$\alpha$ & 20.89 & 98.60\\ 
& $\psi$-FedAvg & 19.60 & --- & FedAvg & 23.39 & --- \\ 
& $\psi$-FADA & 21.62 & 100 & FADA & 22.66 & 99.53\\
& SM & 23.33 & \textbf{19.29} & AvgD & 22.08 & 67.69\\
\midrule[.2em]
\multirow{4}{*}{\rotatebox[origin=c]{90}{{\parbox[c]{9.8mm}{\scriptsize \bf{M//MM}}}}} 
& Ours & \textbf{56.95} & {23.28} & Ours & \textbf{56.95} & \textbf{23.28}\\ 
& Rnd-$\psi$ & 56.20 & 94.87 & Rnd-$\alpha$ & 54.90 & 100\\ 
& $\psi$-FedAvg & 56.12 & --- & FedAvg & 56.36 & --- \\ 
& $\psi$-FADA & 55.30 & 93.79 & FADA & 56.23 & 100\\
& SM & 56.86 & \textbf{19.37} & AvgD & 56.11 & 63.28\\
\midrule
\multirow{4}{*}{\rotatebox[origin=c]{90}{{\parbox[c]{7mm}{\scriptsize \bf{M//U}}}}} 
& Ours & \textbf{86.26} & {25.28} & Ours & \textbf{86.26} & \textbf{25.28}\\ 
& Rnd-$\psi$ & 78.46 & 94.87 & Rnd-$\alpha$ & 78.46 & 99.35\\ 
& $\psi$-FedAvg & 82.25 & --- & FedAvg & 80.39 & --- \\ 
& $\psi$-FADA & 80.65 & 97.68 & FADA & 80.47 & 100\\
& SM & 79.11 & \textbf{19.38} & AvgD & 79.79 & 65.72\\
\midrule
\multirow{4}{*}{\rotatebox[origin=c]{90}{{\parbox[c]{9.8mm}{\scriptsize \bf{MM//U}}}}} 
& Ours & \textbf{54.70} & {23.49} & Ours & \textbf{54.70} & \textbf{23.49}\\ 
& Rnd-$\psi$ & 45.34 & 94.87 & Rnd-$\alpha$ & 46.46 & 99.35\\ 
& $\psi$-FedAvg & 44.94 & --- & FedAvg & 45.24 & --- \\ 
& $\psi$-FADA & 47.92 & 94.23 & FADA & 45.39 & 100\\
& SM & 47.29 & \textbf{19.38} & AvgD & 45.90 & 63.43\\
\bottomrule
\end{tabularx}
\label{tab:alg_nrg_comp_all} 
}
\vspace{-4mm}
\end{table}

\textbf{Effect of Link Formation $\boldsymbol{\alpha}$.} 
{
In Fig.~\ref{fig:mt_alpha_all} and Table~\ref{tab:alg_nrg_comp_all}, we compare ST-LF against the four baselines that determine $\boldsymbol{\alpha}$: 
\begin{enumerate}[label=(\roman*),leftmargin=6ex]
    \item \textit{Random-$\alpha$ (Rnd-$\alpha$)}, which allocates link formation $\boldsymbol{\alpha}$ according to a Dirichlet distribution so that $\sum_{s \in \mathcal{S}}\alpha_{s,t}=1$, 
    \item \textit{FedAvg}~\cite{mcmahan2017communication}, which scales link weights proportionally to sources' local dataset size,
    \item \textit{FADA}~\cite{peng2019federated}, which relies on unsupervised data clustering, distributed adversarial alignment (similar to GAN techniques), and feature disentangling networks to determine the link weights in federated domain adaptation, 
    \item \textit{Avg Degree (AvgD)}, where each source is allocated the average number of links/source from ST-LF, but the specific links and weights are randomly determined.
\end{enumerate}
All baselines rely on ST-LF's $\boldsymbol{\psi}$ determination. 
}

The experiments shown in Fig.~\ref{fig:mt_alpha_all} consist of three kinds of ML dataset manipulation: (i) single ML dataset (e.g., USPS) - all devices draw data from the same underlying dataset, (ii) mixed ML dataset (e.g. M+MM) - all devices draw data from the same underlying dataset composed of multiple datasets from Fig.~\ref{fig:data_snapshot}, (iii) split ML dataset (e.g., U//MM) - devices will draw from different ML datasets (e.g., for U//MM, a source device may draw data from USPS while a target device may contain data from MNIST-M). 
The single dataset experiment is a more classical federated scenario, where devices have non-i.i.d. data distributions drawn from the same ML dataset, but augmented with unlabeled data. 
Meanwhile, the mixed and split dataset experiments are bigger challenges, jointly involving statistical heterogeneity (i.e., non-i.i.d. data distributions) and dataset heterogeneity. 

{\color{black}
ST-LF outperforms the $\alpha$-baselines substantially in all single dataset cases: by $>$9\% for MNIST, and $>$5\% for both USPS and MNIST-M in terms of average accuracy.} 
The mixed and split dataset evaluations also show consistent and significant improvements, though more marginal for cases involving combinations of MNIST-M with other datasets, which is expected given the more challenging nature of mixed and split dataset scenarios.
The importance is that ST-LF maintains its significant performance benefits over all baselines in all cases, highlighting its effectiveness at determining model/link weights $\boldsymbol{\alpha}$. 

{
In Table~\ref{tab:alg_nrg_comp_all}, we also see substantial improvements in network-wide communication energy consumption compared to the multi-source baselines, due to the fact that ST-LF forms fewer source-target links. 
Despite forming fewer source-to-target links, our method achieves the best accuracies and the largest energy consumption savings relative to all other baselines for model/link weights $\boldsymbol{\alpha}$. 
Mixed and split dataset energy consumption measurements in Table~\ref{tab:alg_nrg_comp_all} preserve the same pattern as those for single dataset experiments. 
Energy measurements for server dependent baselines (i.e., FedAvg for $\boldsymbol{\alpha}$ baselines and $\psi$-FedAvg for $\boldsymbol{\psi}$ baselines) are omitted as they depend on the distance between server and devices. 
Overall, we see significant improvements for model/link weights $\boldsymbol{\alpha}$ even when the baselines employ ST-LF's solution for source/target classification $\boldsymbol{\psi}$, which emphasizes the importance of our joint optimization methodology. 
}

\textbf{Effect of Source/Target Determination $\boldsymbol{\psi}$.} 
{
Next, we investigate the impact of source/target determination $\boldsymbol{\psi}$ by comparing ST-LF against the four baselines that influence $\psi$: 
\begin{enumerate}[label=(\roman*),leftmargin=6ex]
    \setcounter{enumi}{4}
    \item \textit{Random-$\psi$ (Rnd-$\psi$)}, which randomly classifies each device as a source or target and then uses Rnd-$\alpha$ for the model/link weights,
    \item \textit{$\psi$-FedAvg}, which uses the same heuristic method to determine $\boldsymbol{\psi}$ as (i) combined with FedAvg's method for model/link weights $\boldsymbol{\alpha}$,
    \item \textit{$\psi$-FADA}, which uses the same method as (ii) but with FADA~\cite{peng2019federated} replacing FedAvg~\cite{mcmahan2017communication},
    \item \textit{Single Matching (SM)}, which is single source to single target matching inspired by~\cite{zhao2019learning}.
\end{enumerate}
Fig.~\ref{fig:mt_psi_all} compares the resulting target accuracies across all methods while parts of Table~\ref{tab:alg_nrg_comp_all} show the resulting energy consumption for all methods. 
Analogous to Fig.~\ref{fig:mt_alpha_all}, Fig.~\ref{fig:mt_psi_all} compares these methods on three kinds of ML dataset manipulations: (1) single ML dataset in Fig.~\ref{fig:avg_gan_full}, (2) mixed ML datasets in Fig.~\ref{fig:mixed_gan_full}, and (3) split ML datasets in Fig.~\ref{fig:split_gan_full}. 
}

{\color{black}
The single dataset experiment, Fig.~\ref{fig:avg_gan_full}, effectively compares ST-LF against the baselines in the classical federated scenario, where devices have non-i.i.d. data distributions albeit combined with the presence of unlabeled data. 
In this setting, ST-LF obtains significant performance improvements over the $\boldsymbol{\psi}$-baselines, specifically by $>8\%$ on MNIST, and $>4\%$ on both USPS and MNIST-M, demonstrating its ability to select the most optimal sources in the classical federated case.  
}

On the other hand, the mixed dataset case in Fig.~\ref{fig:mixed_gan_full} and the split dataset case in Fig.~\ref{fig:split_gan_full} feature multiple ML datasets in the same network environment, thus augmenting the classical federated challenge of statistical heterogeneity (i.e., non-i.i.d. data distributions) with dataset heterogeneity. 
In this more challenging problem, ST-LF maintains a clear and consistent performance advantage relative to all baselines, demonstrating its ability to jointly consider statistical similarity as well as dataset similarity when determining the optimal set of source devices. 
{
In Table~\ref{tab:alg_nrg_comp_all}, ST-LF continues to demonstrate the best accuracies relative to other baselines. 
One heuristic baseline, single matching (SM), does provide incremental energy savings over ST-LF, but leads to substantially worse ML performance across all test cases. 
Furthermore, because it is a one-to-one matching, the performance of SM is highly variable across different settings 
(e.g., the $20\%$ drop in performance relative to ST-LF on MNIST). 
Overall, when considering joint design of source/target classification, source-to-target model/link weights, and energy resource efficiency, ST-LF continues to show a commanding performance. 
}

\section{Conclusion}

\noindent 
We investigated decentralized FL in settings where devices have partially labeled datasets of varying quality. 
This challenging problem augments standard federated settings, which are known for communication heterogeneity and statistical heterogeneity (i.e., non-i.i.d. data distributions) across network devices, by introducing heterogeneous quantities and distributions of unlabeled data across network devices and integrating combinations of unique/independent ML datasets into the standard FL problem. 
We addressed this problem through a novel methodology, ST-LF, for multi-source to multi-target federated domain adaptation. 
In developing ST-LF, we obtained theoretical results for domain generalization errors which are measurable in real-world systems. 
Based on these results, we formulated a concrete optimization problem that jointly (i) determines the optimal source/target classification of devices, (ii) obtains link weights (model combination weights) to match sources to targets, and (iii) considers communication efficiency of model transmissions. 
We showed that ST-LF belongs to a class of mixed-integer signomial programs, which are NP-hard and non-convex, and developed an iterative method to solve it based on refining convex inner approximations. 
Finally, we demonstrated both performance and energy efficiency improvements obtained by ST-LF relative to baselines on common domain adaptation datasets. 


\bibliographystyle{IEEEtran}
\bibliography{References}

\vspace{-10mm}
\begin{IEEEbiographynophoto}
{(Henry) Su Wang} received the Ph.D. in ECE from Purdue University.
\end{IEEEbiographynophoto}

\vspace{-10mm}
\begin{IEEEbiographynophoto}
{Seyyedali Hosseinalipour} is an assistant professor of EE at the University at Buffalo (SUNY).
\end{IEEEbiographynophoto}

\vspace{-10mm}
\begin{IEEEbiographynophoto}
{Christopher G. Brinton (S’08, M’16, SM’20)} is the Elmore Rising Star Assistant Professor of ECE at Purdue University. 
\end{IEEEbiographynophoto}



\newpage
\clearpage
\begingroup
\let\clearpage\relax 
\onecolumn
\appendices


\section{Proof of Theorem~\ref{thm:upper_bound_tte}}
\label{app_ssec:thm2} 

\ourgenbound*
\begin{proof}
Leveraging $\vert \varepsilon_{\mathcal{D}}(h,h^{'}) - \varepsilon_{\mathcal{D}^{'}}(h,h^{'}) \vert \leq \frac{1}{2} d_{\mathcal{H}\Delta \mathcal{H}}(\mathcal{D},\mathcal{D}^{'})$ (Lemma 3~\cite{ben2010theory})
where $h,h^{'} \in \mathcal{H}$ and $\mathcal{D},\mathcal{D}^{'}$ are underlying data distributions, we can upper bound the true target error when the hypothesis function is a weighted combination from multiple source domains, i.e., upper bounding $\varepsilon^{\mathsf{T}}_t(h^{\mathsf{T}}_t,f^{\mathsf{T}}_t)$ for an arbitrary target domain $j$. Thus:
\begin{align}
    & \varepsilon^{\mathsf{T}}_t(h^{\mathsf{T}}_t,f^{\mathsf{T}}_t) = \sum_{s \in \mathcal{S}} \alpha_{s,t} \varepsilon^{\mathsf{T}}_t(h^{\mathsf{T}}_t,f^{\mathsf{T}}_t) 
    \overset{(a)}{\leq} \sum_{s \in \mathcal{S}} \alpha_{s,t} \bigg[ \varepsilon^{\mathsf{T}}_t(h^{\mathsf{T}}_t,h^{\mathsf{S}}_s) + \varepsilon^{\mathsf{T}}_t(h^{\mathsf{S}}_s,f^{\mathsf{T}}_t) \bigg] \\
    & \overset{(b)}{\leq} \sum_{s \in \mathcal{S}} \alpha_{s,t} \bigg[ \varepsilon^{\mathsf{T}}_t(h^{\mathsf{T}}_t,h^{\mathsf{S}}_s) + \varepsilon^{\mathsf{T}}_t(h^{\mathsf{S}}_s,f^{\mathsf{S}}_s) + \varepsilon^{\mathsf{T}}_t(f^{\mathsf{S}}_s,f^{\mathsf{T}}_t) \bigg] \\
    & \overset{(c)}{\leq} \sum_{s \in \mathcal{S}} \alpha_{s,t} \bigg[ \varepsilon^{\mathsf{T}}_t(h^{\mathsf{T}}_t,h^{\mathsf{S}}_s) + \varepsilon^{\mathsf{T}}_t(h^{\mathsf{S}}_s,f^{\mathsf{S}}_s) + \varepsilon^{\mathsf{T}}_t(f^{\mathsf{S}}_s,f^{\mathsf{T}}_t) 
    + \varepsilon^{\mathsf{S}}_s(h^{\mathsf{S}}_s,f^{\mathsf{S}}_s) - \varepsilon^{\mathsf{S}}_s(h^{\mathsf{S}}_s,f^{\mathsf{S}}_s) \bigg] \\ 
    & \overset{(d)}{\leq} \sum_{s \in \mathcal{S}} \alpha_{s,t} \bigg[ \varepsilon^{\mathsf{T}}_t(h^{\mathsf{T}}_t,h^{\mathsf{S}}_s) + \varepsilon^{\mathsf{T}}_t(f^{\mathsf{S}}_s,f^{\mathsf{T}}_t) +\varepsilon^{\mathsf{S}}_s(h^{\mathsf{S}}_s,f^{\mathsf{S}}_s) 
    +\vert \varepsilon^{\mathsf{T}}_t(h^{\mathsf{S}}_s,f^{\mathsf{S}}_s) - \varepsilon^{\mathsf{S}}_s(h^{\mathsf{S}}_s,f^{\mathsf{S}}_s) \vert \bigg] \\
    & \overset{(e)}{\leq} \sum_{s \in \mathcal{S}} \alpha_{s,t} \bigg[ \varepsilon^{\mathsf{T}}_t(h^{\mathsf{T}}_t,h^{\mathsf{S}}_s) + \varepsilon^{\mathsf{T}}_t(f^{\mathsf{S}}_s,f^{\mathsf{T}}_t) +\varepsilon^{\mathsf{S}}_s(h^{\mathsf{S}}_s,f^{\mathsf{S}}_s) 
    + \frac{1}{2}d_{\mathcal{H}\Delta \mathcal{H}}(\mathcal{D}^{\mathsf{T}}_t,\mathcal{D}^{\mathsf{S}}_s) \bigg],
\end{align}
where $(a)$ and $(b)$ are due to the triangle inequality for classification error, $(c)$ introduces $\varepsilon^{\mathsf{S}}_s(h^{\mathsf{S}}_s,f^{\mathsf{S}}_s)$, $(d)$ takes the magnitude of the difference between $\varepsilon^{\mathsf{T}}_t(h^{\mathsf{S}}_s,f^{\mathsf{S}}_s)$ and $\varepsilon^{\mathsf{S}}_s(h^{\mathsf{S}}_s,f^{\mathsf{S}}_s)$, and $(e)$ uses Lemma 3 from~\cite{ben2010theory}.
\end{proof}

\newpage 

\section{Proof of Lemma~\ref{lemma1}}
\label{app_ssec:lemma1}

\hcombination*

\begin{proof}
By definition:
\begin{align}
    & \varepsilon^{\mathsf{T}}_t(h^{\mathsf{T}}_t,h^{\mathsf{S}}_s) = \mathbb{E}_{x \sim \mathcal{D}^{\mathsf{T}}_t}[\vert h^{\mathsf{T}}_t(x) - h^{\mathsf{S}}_s(x) \vert].
\end{align}
Considering the right hand side, $\mathbb{E}_{x \sim \mathcal{D}_{T_j}}[\vert h_j(x) - h_i(x) \vert ]$, we have
\begin{align}
    & \mathbb{E}_{x \sim \mathcal{D}^{\mathsf{T}}_t}[\vert h^{\mathsf{T}}_t(x) - h^{\mathsf{S}}_s(x) \vert] 
    \leq \frac{1}{{\widehat{D}}^{\mathsf{T}}_t} \sum_{k=1}^{\widehat{D}^{\mathsf{T}}_t} \vert h^{\mathsf{T}}_t(x_k) - h^{\mathsf{S}}_s(x_k)\vert  
    + 2 Rad_{\widehat{\mathcal{D}}^{\mathsf{T}}_t}(\vert \mathcal{H}^{\mathsf{T}}_{t} - \mathcal{H}\vert) + 3 \sqrt{\frac{log(2/\delta)}{2\widehat{D}^{\mathsf{T}}_t}} \label{eq:rad_begin}, 
\end{align}
by~\cite{bartlett2002rademacher}, and where we define $\mathcal{H}^{\mathsf{T}}_t$ as the hypothesis space derived from the weighted averages of hypothesis from $\mathcal{H}$. 

We now consider $Rad_{\widehat{\mathcal{D}}^{\mathsf{T}}_t}(\mathcal{H}^{\mathsf{T}}_t - \mathcal{H})$:
\begin{align}
    & Rad_{\widehat{\mathcal{D}}^{\mathsf{T}}_t}(\mathcal{H}^{\mathsf{T}}_t - \mathcal{H}) 
    = \mathbb{E}_{\sigma}\left[ \sup_{h^{'} \in \mathcal{H}^{\mathsf{T}}_t - \mathcal{H}} \frac{1}{\widehat{D}^{\mathsf{T}}_t} \sum_{k=1}^{\widehat{D}^{\mathsf{T}}_t} \sigma_k h^{'}(x_k) \right] \\
    & \overset{(a)}{=} \mathbb{E}_{\sigma} \left[\sup_{h^{\mathsf{S}}_l,h^{\mathsf{S}}_s, \in \mathcal{H}, \forall l} \frac{1}{\widehat{D}^{\mathsf{T}}_t} \sum_{k=1}^{\widehat{D}^{\mathsf{T}}_t} \sigma_k 
    (\sum_{l \in \mathcal{N}} \alpha_{l,t} h^{\mathsf{S}}_l(x_k) - h^{\mathsf{S}}_s(x_k)) \right] \\
    & \overset{(b)}{=} \mathbb{E}_{\sigma} \left[\sup_{h^{\mathsf{S}}_l \in \mathcal{H}, \forall l} \frac{1}{\widehat{D}^{\mathsf{T}}_t} \sum_{k=1}^{\widehat{D}^{\mathsf{T}}_t} \sigma_k  \sum_{l \in \mathcal{N}}\alpha_{l,j} h^{\mathsf{S}}_l(x_k)\right] 
    + \mathbb{E}_{\sigma}\left[\sup_{h^{\mathsf{S}}_s \in \mathcal{H}} \frac{1}{\widehat{D}^{\mathsf{T}}_t} \sum_{k=1}^{\widehat{D}^{\mathsf{T}}_t} \sigma_k  h^{\mathsf{S}}_s(x_k) \right]  \\
    & \overset{(c)}{=} \sum_{l \in \mathcal{N}} \alpha_{l,j} \mathbb{E}_{\sigma}\left[\sup_{h^{\mathsf{S}}_l \in \mathcal{H}, \forall l} \frac{1}{\widehat{D}_{T_j}} \sum_{k=1}^{\widehat{D}_{T_j}} \sigma_k   h^{\mathsf{S}}_l(x_k)\right] + Rad_{\widehat{\mathcal{D}}^{\mathsf{T}}_t}(\mathcal{H}) 
    \overset{(d)}{=} 2 Rad_{\widehat{\mathcal{D}}^{\mathsf{T}}_t}(\mathcal{H})
    \label{eq:rad_result},
\end{align}
where $(a)$ is due to $h' = h^{\mathsf{T}}_t - h^{\mathsf{S}}_s$ and the definition of $h^{\mathsf{T}}_t = \sum_{l\in\mathcal{N}} \alpha_{l,t} h^{\mathsf{S}}_l$, $(b)$ is due to linearity of the supremum $\sup(A-B) = \sup(A) + \sup(-B)$ and expectation operators, and $\sigma_k$ absorbing the negative sign on $h^{\mathsf{S}}_s$, and $(c)$ is from the linearity of expectation of supremum and the definition of empirical Rademacher complexity, and $(d)$ applies the definition of empirical Rademacher complexity again. Next, by the Ledoux-Talagrand's contraction lemma\footnote{See "Probability in Banach Spaces: isoperimetry and processes" by Michel Ledoux and Michel Talagrand.}, we have
\begin{equation} \label{eq:lt_lemma}
    Rad_{\widehat{D}^{\mathsf{T}}_t}(\vert \mathcal{H}^{\mathsf{T}}_t - \mathcal{H} \vert ) \leq Rad_{\widehat{D}^{\mathsf{T}}_t}(\mathcal{H}^{\mathsf{T}}_t - \mathcal{H}),
\end{equation}
as $\vert \cdot \vert$ is $1-$Lipschitz. Combining~\eqref{eq:rad_begin},~\eqref{eq:rad_result}, and~\eqref{eq:lt_lemma} yields: 
\begin{align}
    & \mathbb{E}_{x \sim \widehat{\mathcal{D}}_{T_j}}[\vert h^{\mathsf{T}}_t(x) - h^{\mathsf{S}}_s(x) \vert] \leq \frac{1}{\widehat{D}^{\mathsf{T}}_t} \sum_{k=1}^{\widehat{D}^{\mathsf{T}}_t} \vert h^{\mathsf{T}}_t(x_k) - h^{\mathsf{S}}_s(x_k) \vert 
    + 4 Rad_{\widehat{\mathcal{D}}^{\mathsf{T}}_t}(\mathcal{H}) + 3 \sqrt{\frac{log(2/\delta)}{2\widehat{D}^{\mathsf{T}}_t}},
\end{align}
which concludes the proof.
\end{proof}
\newpage 

\section{Proof of Corollary~\ref{coro:gen_error_bound_rademacher}}
\label{app_ssec:corollary1}

\boundrad*

\begin{proof}
Combining the results of Theorem~\ref{thm:upper_bound_tte} and Lemma~\ref{lemma1} yields
\begin{align} \label{cor1_init_state}
    & \varepsilon^{\mathsf{T}}_t(h^{\mathsf{T}}_t) \leq \sum_{i \in \mathcal{N}} \alpha_{i,j} \Bigg[ \widehat{\varepsilon}^{\mathsf{S}}_s(h^{\mathsf{S}}_s) + 2 Rad_{\widehat{\mathcal{D}}^{\mathsf{S}}_s}(\mathcal{H}) + 3 \sqrt{\frac{\log(2/\delta)}{2\widehat{D}^{\mathsf{S}}_s}}  + \frac{1}{2} {d}_{\mathcal{H}\Delta\mathcal{H}}({\mathcal{D}}^{\mathsf{T}}_t,{\mathcal{D}}^{\mathsf{S}}_s) \\ \nonumber
    & + \varepsilon^{\mathsf{T}}_t(f^{\mathsf{T}}_t,f^{\mathsf{S}}_s) + 
    \widehat{\varepsilon}^{\mathsf{T}}_t(h^{\mathsf{T}}_t,h^{\mathsf{S}}_s) + 4 Rad_{\widehat{\mathcal{D}}^{\mathsf{T}}_t}(\mathcal{H}) + 3 \sqrt{\frac{\log(2/\delta)}{2\widehat{D}^{\mathsf{T}}_t}} \Bigg] . 
\end{align}
Next, we upper bound the true source-target hypothesis divergence $d_{\mathcal{H}\Delta\mathcal{H}}(\mathcal{D}^{\mathsf{T}}_{t},\mathcal{D}^{\mathsf{S}}_{s})$ by the empirical source-target hypothesis divergence $\widehat{d}_{\mathcal{H}\Delta\mathcal{H}}(\widehat{\mathcal{D}}^{\mathsf{T}}_t,\widehat{\mathcal{D}}^{\mathsf{S}}_s)$ as follows:
\begin{align} \label{div2rad}
    & d_{\mathcal{H}\Delta\mathcal{H}}(\mathcal{D}^{\mathsf{T}}_t,\mathcal{D}^{\mathsf{S}}_s) \overset{(i)}{\leq} d_{\mathcal{H}\Delta\mathcal{H}}(\mathcal{D}^{\mathsf{T}}_t,\widehat{\mathcal{D}}^{\mathsf{T}}_t) 
    + d_{\mathcal{H}\Delta\mathcal{H}}(\widehat{\mathcal{D}}^{\mathsf{T}}_t,\widehat{\mathcal{D}}^{\mathsf{S}}_s) + d_{\mathcal{H}\Delta\mathcal{H}}(\widehat{\mathcal{D}}^{\mathsf{S}}_s,\mathcal{D}^{\mathsf{S}}_s)  \nonumber \\ 
    & \overset{(ii)}{\leq} d_{\mathcal{H}\Delta\mathcal{H}}(\widehat{\mathcal{D}}^{\mathsf{T}}_t,\widehat{\mathcal{D}}^{\mathsf{S}}_s) 
    + 4 (Rad_{\widehat{\mathcal{D}}^{\mathsf{S}}_s}(\mathcal{H})+Rad_{\widehat{\mathcal{D}}^{\mathsf{T}}_t}(\mathcal{H})) + 6 \bigg( \sqrt{\frac{\log(2/\delta)}{2\widehat{D}^{\mathsf{S}}_s}} + \sqrt{\frac{\log(2/\delta)}{2\widehat{D}^{\mathsf{T}}_t}} \bigg),
\end{align}
where $(i)$ results from repeated application of the triangle property of hypothesis divergences, and 
$(ii)$ results from the following rearrangement of Bartlett and Mendelson's Lemma~\cite{bartlett2002rademacher}, which is reproduced in Lemma~\ref{lemma:bartlett} of Appendix~\ref{app_ssec:bartlett_mendelson}. 
When the distribution divergence is compared over empirical distributions, $\widehat{d}_{\mathcal{H}\Delta\mathcal{H}}(\widehat{\mathcal{D}}^{\mathsf{T}}_t,\widehat{\mathcal{D}}^{\mathsf{S}}_s)$ and ${d}_{\mathcal{H}\Delta\mathcal{H}}(\widehat{\mathcal{D}}^{\mathsf{T}}_t,\widehat{\mathcal{D}}^{\mathsf{S}}_s)$ are interchangeable, but we use $\widehat{d}_{\mathcal{H}\Delta\mathcal{H}}(\widehat{\mathcal{D}}^{\mathsf{T}}_t,\widehat{\mathcal{D}}^{\mathsf{S}}_s)$ for clarity in the main manuscript. 

{\color{black}In order to get the necessary rearrangement of Lemma~\ref{lemma:bartlett}, we start from the definition of $d_{\mathcal{H}\Delta\mathcal{H}}$.
Here, we have 
\begin{equation} \label{eq:d_HDH_original}
    d_{\mathcal{H}\Delta\mathcal{H}}(\mathcal{D},\widehat{\mathcal{D}}) 
    = 
    2 \sup_{g \in \mathcal{H}\Delta\mathcal{H}} \vert Pr_{x \in \mathcal{D}}[I(g(x))] - Pr_{\widehat{x} \in \widehat{\mathcal{D}}}[I(g(\widehat{x}))]  \vert,
\end{equation}
where $I(\cdot)$ is the indicator function, and we use $x$ to denote data in $\mathcal{D}$ and $\widehat{x}$ to denote data in $\widehat{\mathcal{D}}$. Since the probabilities are over the datasets,~\eqref{eq:d_HDH_original} can be equivalently written as: 
\begin{equation} \label{eq:d_HDH_expect} 
d_{\mathcal{H}\Delta\mathcal{H}}(\mathcal{D},\widehat{\mathcal{D}}) 
= 
2 \sup_{g \in \mathcal{H}\Delta\mathcal{H}} 
\Bigg\vert \frac{1}{D} \sum_{i=1}^{D} g(x_i) - \frac{1}{\widehat{D}} \sum_{i=1}^{\widehat{D}} g(\widehat{x}_i) \Bigg\vert 
= 
2 \sup_{g \in \mathcal{H}\Delta\mathcal{H}} 
\bigg\vert \mathbb{E}_{x \in \mathcal{D}}[I(g(x))] - \mathbb{E}_{\widehat{x} \in \widehat{\mathcal{D}}}[I(g(\widehat{x}))] \bigg\vert.
\end{equation}
Additionally, we can rearrange Bartlett and Mendelson's Lemma (Lemma~\ref{lemma:bartlett} in Appendix~\ref{app_ssec:bartlett_mendelson}) to obtain the following one-sided bound:
\begin{equation} \label{eq:bm_lemma_rearranged}
    \mathbb{E}_{x \in \mathcal{D}} [I(g(x))] - \frac{1}{\widehat{D}} \sum_{i=1}^{\widehat{D}} g(\widehat{x}_i) = \mathbb{E}_{x \in \mathcal{D}} [I(g(x))] - \mathbb{E}_{\widehat{x} \in \widehat{\mathcal{D}}} [I(g(\widehat{x}))]
    \leq 2 Rad_{\widehat{\mathcal{D}}}(\mathcal{H}) + 3 \sqrt{\frac{\log(2/\delta)}{2\widehat{D}}}. 
\end{equation}
Finally, taking the absolute value and $2\sup_{g \in \mathcal{H}\Delta\mathcal{H}}$ from both sides of~\eqref{eq:bm_lemma_rearranged} then combining with~\eqref{eq:d_HDH_expect} yields:
    \begin{equation} \label{eq:bm_lemma_2_d_HDH}
        2 \sup_{g \in \mathcal{H}\Delta\mathcal{H}} 
        \bigg\vert
        \mathbb{E}_{x \in \mathcal{D}} [I(g(x))] - \mathbb{E}_{\widehat{x} \in \widehat{\mathcal{D}}} [ I(g(\widehat{x}))] 
        \bigg\vert
        = 
        d_{\mathcal{H}\Delta\mathcal{H}}(\mathcal{D},\widehat{\mathcal{D}}) 
        \leq 
        4 Rad_{\widehat{\mathcal{D}}}(\mathcal{H}) + 6 \sqrt{\frac{\log(2/\delta)}{2\widehat{D}}}.
    \end{equation}
}

\noindent Thus, we can combine~\eqref{cor1_init_state} and~\eqref{div2rad}, and subsequently apply Lemma~\ref{lemma:massart_lemma} in Appendix~\ref{app_ssec:massart_lemma} to bound the empirical Rademacher complexities to obtain the result.
\end{proof}
\newpage 

\section{Massart's Lemma} 
\label{app_ssec:massart_lemma}


{\color{black}
\begin{lemma}\label{lemma:massart_lemma}
Finite Class Lemma (Massart's Lemma)~\cite{massartlecture, shalev2014understanding}.
Given some $\boldsymbol{Q} = \{x_i\}_{i=1}^{n}$, let $\boldsymbol{A}_{\boldsymbol{Q}} = \{ \boldsymbol{a}_1,\cdots,\boldsymbol{a}_{N} \}$, where $\boldsymbol{a}_j = \{h_j(x_1),\cdots,h_j(x_n) \}$, $h_j \in \mathcal{H}$, and $\mathcal{H}$ is a binary hypothesis space. 
Define $\bar{\boldsymbol{a}} = \frac{1}{N} \sum_{j=1}^{N} \boldsymbol{a}_j$ and $N = \vert \boldsymbol{A}_{\boldsymbol{Q}} \vert \leq 2^{n}$. 
Then, 
\begin{equation}\label{eq:massart_statement_eq}
    Rad_{\boldsymbol{Q}}(\mathcal{H}) \leq \max_{\boldsymbol{a} \in \boldsymbol{A}_{\boldsymbol{Q}}} \Vert \boldsymbol{a} - \bar{\boldsymbol{a}} \Vert_2 \frac{\sqrt{2\log(N)}}{n}.
\end{equation}
\end{lemma}

Since $\mathcal{H}$ is a binary hypothesis space, the most that $\boldsymbol{a},\bar{\boldsymbol{a}} \in \boldsymbol{A}_{\boldsymbol{Q}}$ can differ in a single position is by $1$. In other words, using $\boldsymbol{a}_k$ and $\bar{\boldsymbol{a}}_k$ to denote the $k$-th position of $\boldsymbol{a}$ and $\bar{\boldsymbol{a}}$, we have $\vert \boldsymbol{a}_k - \bar{\boldsymbol{a}}_k \vert \leq 1$ for all $k \in \{1,\cdots,n\}$ in $\boldsymbol{Q}$. 
We can then use this to upper bound~\eqref{eq:massart_statement_eq} as $\max_{\boldsymbol{a} \in \boldsymbol{A}_{\boldsymbol{Q}}} \Vert \boldsymbol{a} - \bar{\boldsymbol{a}} \Vert_{2} \leq \sqrt{\sum_{i \in \{1,\cdots,n\}} 1^{2}} = \sqrt{n}$. 
Thus,~\eqref{eq:massart_statement_eq} can also be bounded for binary $\mathcal{H}$ as: 
\begin{equation} \label{eq:ml_eq2}
    Rad_{\boldsymbol{Q}}(\mathcal{H}) \leq \sqrt{\frac{2\log(N)}{n}}.
\end{equation}

Consequently, if $\mathcal{H}$ shatters $\boldsymbol{Q}$, it means that $\mathcal{H}$ contains all possible hypothesis $h$ that map $\boldsymbol{Q}$ to $\{0,1\}$. In other words, if shattering occurs, then $N = \vert \mathcal{H} \circ \boldsymbol{Q} \vert = 2^{n}$ and~\eqref{eq:ml_eq2} becomes:
\begin{equation} \label{eq:ml_eq3}
    Rad_{\boldsymbol{Q}}(\mathcal{H}) \leq \sqrt{\frac{2\log(N)}{n}} \leq \sqrt{\frac{2\log(2^{n})}{n}} = \sqrt{2\log{2}}.
\end{equation}

}

\section{Bartlett and Mendelson's Lemma}
\label{app_ssec:bartlett_mendelson}

\begin{lemma}~\cite{bartlett2002rademacher}. 
\label{lemma:bartlett}
Let $\mathcal{H} \subseteq [0,1]^{\mathcal{X}}$, where $\mathcal{X}$ is an input space. Then, $\forall \delta > 0$, with probability at least $1-\delta$, the following inequality holds $\forall h \in \mathcal{H}$:
\begin{equation}
    \mathbb{E}[h(x)] \leq \frac{1}{n} \sum_{i=1}^{n} h(x_i) + 2 Rad_{\boldsymbol{Q}}(\mathcal{H}) + 3 \sqrt{\frac{\log(2/\delta)}{2n}},
\end{equation}
where $\boldsymbol{Q} = \{x_i\}_{i=1}^{n}$, and $Rad_{\boldsymbol{Q}}(\mathcal{H})$ is the empirical Rademacher complexity defined as in~\eqref{eq:rad}. 
\end{lemma}

\section{A Discussion on Federated Divergence Estimation}
\label{app_ssec:divergence_estimation}
{\color{black} Similar to popular works in the field~\cite{ben2010theory,zhao2018adversarial}, we compute the empirical $\mathcal{H}$ divergence in place of the empirical $\mathcal{H}\Delta\mathcal{H}$ divergence.}
{\color{black}We can adjust the definition of the true $d_{\mathcal{H}}$ divergence from Definition~\ref{def:div} for empirical $\widehat{d}_{\mathcal{H}}$ divergences by substituting true data terms with their empirical counterparts. Thus, we obtain:  $\widehat{d}_{\mathcal{H}} = 2 \sup_{A \in \mathcal{A}_{\mathcal{H}}} \vert Pr_{\widehat{\mathcal{D}}^{\mathsf{T}}_t}(A) - Pr_{\widehat{\mathcal{D}}^{\mathsf{S}}_s}(A) \vert$, which means that the determination of empirical distribution divergence is equivalent to quantifying the separability of source and target domains based on the hypothesis space $\mathcal{H}$.
} 
{\color{black}Here, the researchers in~\cite{ben2010theory} propose a methodology to compute the empirical distribution divergence $\widehat{d}_{\mathcal{H}}(\widehat{\mathcal{D}}^{\mathsf{T}}_t,\widehat{\mathcal{D}}^{\mathsf{S}}_s)$ for symmetric hypothesis spaces, such as ours, in which if $h \in \mathcal{H}$ then $1-h \in \mathcal{H}$ as well.}
Their method computes $\widehat{d}_{\mathcal{H}}(\widehat{\mathcal{D}}^{\mathsf{T}}_t,\widehat{\mathcal{D}}^{\mathsf{S}}_s)$ by relabeling source data $\widehat{\mathcal{D}}^{\mathsf{S}}_s$ as $0$ and target data $\widehat{\mathcal{D}}^{\mathsf{T}}_t$ as $1$, and subsequently training a binary classifier based on the hypothesis space $\mathcal{H}$ to distinguish between domains. However, such methods assume that both source and target domains' data are located in the same location. 
In a federated setting, where data privacy becomes a major concern and data sharing is discouraged, we cannot readily apply the method of~\cite{ben2010theory}. We consequently develop a decentralized peer-to-peer algorithm, Algorithm~\ref{alg:fed_div_est}, to compute empirical distribution divergence among pairs of network devices. Our method only relies on hypothesis exchange among the devices and eliminates the need for any data transfer across the network. 

\newpage 

\section{{\color{black}Tightness/Looseness of Theoretical Results}}
\label{app_sec:tl_bounds}
To investigate the tightness/looseness of the proposed generalization error bounds, we have conducted additional simulations to measure bounds in Table~\ref{tab:bound_tl}. 
Note that some of the bounds contain terms on the true error, which would require having the true set of all data that a device could possibly encounter and/or generate in order to calculate. As a result, when true error terms appear, we estimate their values in Table~\ref{tab:bound_tl} using their corresponding empirical error terms. 
\\

Since the left hand side (L.H.S.) for both Theorem~\ref{thm:upper_bound_tte} and Corollary~\ref{coro:gen_error_bound_rademacher} is the same true error at target devices term, we use a single row to represent both cases in Table~\ref{tab:bound_tl}. The last two rows measure the empirical estimates for the right hand side (R.H.S.) of Theorem 2 and Corollary 1. 
Comparing the LHS and RHS for Theorem~\ref{thm:upper_bound_tte} (i.e., the first two rows of Table~\ref{tab:bound_tl}, we can see that the difference tends to be small, roughly from $0$ to $2\text{x}$.
The difference between LHS and RHS for Corollary~\ref{coro:gen_error_bound_rademacher} are larger however, as Corollary~\ref{coro:gen_error_bound_rademacher} can be viewed as an upper bound for Theorem~\ref{thm:upper_bound_tte} owing to converting the intractable terms in the theorem (i.e., the true error and divergence terms) to empirically measurable quantities. 
{\color{black}
In this process, additional quantities, which bounded empirical Rademacher complexities via Massart's lemma, were introduced that made the RHS larger than the LHS for Corollary~\ref{coro:gen_error_bound_rademacher}. 
The resulting bound in Corollary~\ref{coro:gen_error_bound_rademacher} considers possible shattering. 
Therefore, Corollary~\ref{coro:gen_error_bound_rademacher} yields the worst-case bound, making it applicable to a broad category of problems in federated domain adaptation. 
Nonetheless, in the empirical estimates for Corollary~\ref{coro:gen_error_bound_rademacher}, the RHS values differ from the LHS by roughly an order of magnitude, which is reasonable given that the LHS is a true error term that cannot be optimized directly and that Corollary~\ref{coro:gen_error_bound_rademacher} bounds the worst-case scenario.}
\\


\begin{table}[h]
\caption{
{\color{black}
Comparing the left and right sides of Theorem 2 and Corollary 1, i.e.,~\eqref{eq:thm1_main_result} and~\eqref{eq:coro_gen_error_bound_rademacher}. Since true errors (e.g., $\varepsilon^{\mathsf{S}}_{s}(h^{\mathsf{S}}_s)$ or $\varepsilon^{\mathsf{T}}_{t}(h^{\mathsf{T}}_t)$) cannot be calculated in practice, when they appear (e.g., both sides of Theorem 2), we substitute them for empirical error terms instead (e.g., $\widehat{\varepsilon}^{\mathsf{S}}_{s}(h^{\mathsf{S}}_s)$) to assess the tightness/looseness of the bounds. \textbf{L.H.S.} denotes left hand side and \textbf{R.H.S.} denotes right hand side.}}
\begin{tabularx}{0.99\textwidth}
{m{8em} m{4.5em} m{4.5em} m{4.5em} m{4.5em} m{4.5em} m{4.5em} m{4.5em} m{4.5em} m{4.5em}}
\toprule[.2em]
& \multicolumn{3}{c}{\textbf{Single Dataset}} & \multicolumn{3}{c}{\textbf{Mixed Dataset}} & \multicolumn{3}{c}{\textbf{Split Dataset}} \\
\cmidrule(lr){2-4} \cmidrule(lr){5-7} \cmidrule(lr){8-10}
& M & U & MM & M+MM & M+U & MM+U & M//MM & M//U & MM//U \\
\midrule
L.H.S. Thm2/Cor1 & 0.25 & 0.20 & 0.56 & 0.42 & 0.11 & 0.68 & 0.40 & 0.16 & 0.41 \\ 
R.H.S. Thm2 & 0.34 & 0.46 & 0.62 & 0.58 & 0.35 & 0.76 & 0.66 & 0.41 & 0.42 \\
R.H.S. Cor1 & 8.28 & 8.40 & 8.55 & 8.52 & 8.30 & 8.70 & 8.60 & 8.35 & 8.36 \\ 
\bottomrule
\end{tabularx}
\label{tab:bound_tl}
\end{table}
\newpage

\section{Solution to Optimization}
\label{app_ssec:st_lf_solution}
Our solution methodology for $(\boldsymbol{\mathcal{P}})$ involves transforming $(\boldsymbol{\mathcal{P}})$ into a geometric program, which, after a logarithmic change of variables, becomes a convex program. To explain this procedure, we first explain geometric programming in Sec.~\ref{app_sssec:gp} and then explain the details behind our transformations for $(\boldsymbol{\mathcal{P}})$ in Sec.~\ref{app_sssec:tx_2gp}.

\subsubsection{Geometric Programming}
\label{app_sssec:gp}

A standard GP is a non-convex
problem formulated as minimizing a posynomial under posynomial 
inequality constraints and monomial equality constraints~\cite{boyd2007tutorial,chiang2007power}:
\begin{equation}\label{eq:GPformat}
\begin{aligned}
&\min_{\bm{y}} g_0 (\bm{y})~~\\
& \textrm{s.t.} ~~\\ &g_i(\bm{y})\leq 1, \; i=1,\cdots,I,\\& f_\ell(\bm{y})=1, \; \ell=1,\cdots,L,
\end{aligned}
\end{equation}
where $g_i(\bm{y})=\sum_{m=1}^{M_i} d_{i,m} y_1^{\beta^{(1)}_{i,m}} \cdots y_n ^{\beta^{(n)}_{i,m}}$, $\forall i$, 
and $f_\ell(\bm{y})= d_\ell y_1^{\beta^{(1)}_\ell}  \cdots y_n ^{\beta^{(n)}_\ell}$, $\forall \ell$. Since the log-sum-exp function $f(\bm{y}) = \log \sum_{j=1}^n e^{y_j}$ is convex, where $\log$ denotes the natural logarithm, the GP in its standard format can be transformed into a convex programming formulation, via logarithmic change of variables and constants $z_i=\log(y_i)$, $b_{i,k}=\log(d_{i,k})$, $b_\ell=\log (d_\ell)$, and application of $\log$ on the objective and constrains of~\eqref{eq:GPformat}.
The result is as follows:
\begin{equation}~\label{GPtoConvex}
\begin{aligned}
&\min_{\bm{z}} \;\log \sum_{m=1}^{M_0} e^{\left(\bm{\beta}^{\top}_{0,m}\bm{z}+ b_{0,m}\right)}\\&\textrm{s.t.}~ \log \sum_{m=1}^{M_i} e^{\left(\bm{\beta}^{\top}_{i,m}\bm{z}+ b_{i,m}\right)}\leq 0,~ i=1,\cdots,I, \\&~~~~~~ \bm{\beta}_\ell^\top \bm{z}+b_\ell =0,\; \ell=1,\cdots,L,
\end{aligned}
\end{equation}
where $\bm{z}=[z_1,\cdots,z_n]^\top$, $\bm{\beta}_{i,m}=\left[\beta_{i,m}^{(1)},\cdots, \beta_{i,m}^{(n)}\right]^\top$, $\forall i,m$, and $\bm{\beta}_{\ell}=\left[\beta_{\ell}^{(1)},\cdots, \beta_{\ell}^{(n)}\right]^\top$\hspace{-2mm}, $\forall \ell$.

\subsubsection{Optimization Problem Transformation}
\label{app_sssec:tx_2gp}
With the preliminaries on geometric programming, we can revisit $(\boldsymbol{\mathcal{P}})$ and transform it from a signomial program to a geometric program. For explanatory clarity, we write $(\boldsymbol{\mathcal{P}})$ with the full expressions for all terms below:
\begin{align}
    & (\boldsymbol{\mathcal{P}}):~\argmin_{\boldsymbol{\alpha},\boldsymbol{\psi}} 
    \phi^{\mathsf{S}}  \underbrace{\sum_{i \in \mathcal{N}} (1-\psi_i) S_i}_{(a)}
    + \phi^{\mathsf{T}}  \underbrace{\sum_{j \in \mathcal{N}} \psi_j \sum_{i \in \mathcal{N}} (1-\psi_i) \alpha_{i,j} T_{i,j}
    }_{(b)}
    + \phi^{\mathsf{E}} \underbrace{ \sum_{i \in \mathcal{N}} \sum_{j \in \mathcal{N}} E_{i,j} (\alpha_{i,j})}_{(c)} \label{app_eq:obj_fxn_full}\\[-1em]
    & \textrm{s.t.} \nonumber \\
    & h_j^{\mathsf{T}} = \sum_{i \in \mathcal{N}} \alpha_{i,j} (1-\psi_i) \psi_j h_i^{\mathsf{S}}, \forall j \in \mathcal{N} \label{app_eq:target_hypothesis_def} \\
    & \sum_{i \in \mathcal{N}} \alpha_{i,j} = \psi_j, \forall j \in \mathcal{N} \label{app_eq:only_targets_receive} \\
    & E_{i,j}(\alpha_{i,j}) = K_{i,j} \frac{\alpha_{i,j}}{\alpha_{i,j}+\epsilon_{E}}, \forall i,j \in \mathcal{N} \label{app_eq:def_E_i} \\ 
    & 0 \leq \alpha_{i,j} \leq 1, \psi_i(t) \in \{ 0,1 \}, \forall i,j \in \mathcal{N} \label{app_eq:alpha_psi_limits},
\end{align}
where $K_{i,j}$ is determined in Sec.~\ref{sec:experiments}. Recall that in the main text, we defined 
{\color{black} $S_i = \left( \widehat{\varepsilon}^{\mathsf{S}}_i(h^{\mathsf{S}}_i)+2\sqrt{2\log(2)}
+ 3\sqrt{\frac{\log(2/\delta)}{2\widehat{D}^{\mathsf{S}}_i}}  \right)$, 
and 
$T_{i,j} = 
\bigg( \widehat{\varepsilon}^{\mathsf{S}}_i(h^{\mathsf{S}}_i)+10\sqrt{2\log(2)}
+ \varepsilon^{\mathsf{T}}_j(f^{\mathsf{T}}_j,f^{\mathsf{S}}_i)
+ \frac{1}{2}\widehat{d}_{\mathcal{H}\Delta\mathcal{H}}(\widehat{\mathcal{D}}^{\mathsf{T}}_j,\widehat{\mathcal{D}}^{\mathsf{S}}_i)
+ \widehat{\varepsilon}^{\mathsf{T}}_j(h^{\mathsf{T}}_j,h^{\mathsf{S}}_i) +6 \left( \sqrt{ \log(2/\delta)/(2\widehat{D}^{\mathsf{S}}_i)}
    + \sqrt{ \log(2/\delta)/ (2\widehat{D}^{\mathsf{T}}_j)} \right) \bigg)$. 
}
We will also use $S_i$ and $T_{i,j}$ to make some of our explanations more concise below.

In the following we transform $(\boldsymbol{\mathcal{P}})$ into a GP format, starting with the objective function~\eqref{app_eq:obj_fxn_full}. 

\textbf{Term $(a)$ in Objective Function.} Term $(a)$ involves sums $(1-\psi_i)S_i$ for $i \in \mathcal{N}$,  which is not in the format of GP since it includes the negative optimization variables $-\psi_i$. So, we bound it by an auxiliary variable $\chi^{\mathsf{S}}_i > 0$ as follows:
\begin{align}
    & (1-\psi_i)
    S_i \leq \chi^{\mathsf{S}}_i \label{app_eq:a_obj}\\
    & \frac{1}{\psi_i +
    \frac{\chi^{\mathsf{S}}_i}{S_i}} \leq 1 \label{app_eq:a_obj_adjust},
\end{align}
where we obtain~\eqref{app_eq:a_obj_adjust} by rearranging~\eqref{app_eq:a_obj}. The auxiliary variables $\chi^{\mathsf{S}}_i$ now replace the terms inside the summation of term $(a)$ in the objective function, and those terms become the inequality constraint developed in~\eqref{app_eq:a_obj_adjust}. 
It can be seen that the fraction in~\eqref{app_eq:a_obj_adjust} is not in the format of GP since it is an inequality with a posynomial in the denominator, which is not a posynomial. We thus exploit the arithmetic-geometric mean inequality (Lemma~\ref{Lemma:AG_mean}) to approximate the denominator with a monomial:
\begin{align} \label{app_eq:approx_term_a}
    & F_i(\bm{x}) = \psi_i + \frac{\chi^{\mathsf{S}}_i}{S_i} \geq \widehat{F}_i(\bm{x};\ell) \triangleq \left(\frac{\psi_i F_i([\bm{x}]^{\ell-1}) } {[\psi_i]^{\ell-1}} \right) ^ {\frac{[\psi_i]^{\ell-1}}{F_i([\bm{x}]^{\ell-1})}} 
    \left( \frac{\chi^{\mathsf{S}}_i F_i([\bm{x}]^{\ell-1})}{[\chi^{\mathsf{S}}_i]^{\ell-1}} \right)^ {\frac{[\chi^{\mathsf{S}}_i/ S_i]^{\ell-1}}{F_i([\bm{x}]^{\ell-1})}},
\end{align}
where we use the fact that $S_i$ is a constant to obtain $S_i = [S_i]^{\ell-1}$. 
We finally approximate $(1-\psi_i)S_i$ for each $i \in \mathcal{N}$ in term $(a)$ as follows:
\begin{tcolorbox}[ams align]
& \frac{1}{\widehat{F}_i(\bm{x};\ell)} \leq 1, \forall i \in \mathcal{N}\\
& \chi^{\mathsf{S}}_i \geq 0, \forall i \in \mathcal{N}.
\end{tcolorbox}

\textbf{Term $(b)$ in Objective Function.} 
Term $(b)$ contains optimization variables within $T_{i,j}$. We first use an auxiliary variable for $T_{i,j}$, and then use a separate one for the term inside the summation in $(b)$.
Specifically, within $T_{i,j}$, the hypothesis comparison term $\widehat{\varepsilon}^{\mathsf{T}}_j(h^{\mathsf{T}}_j,h^{\mathsf{S}}_i)$ is defined as $\frac{1}{\widehat{D}^{\mathsf{T}}_j} \sum_{x_d \in \widehat{D}^{\mathsf{T}}_j} \vert h^{\mathsf{T}}_{j}(x_d) - h^{\mathsf{S}}_i(x_d)\vert$, which relies on optimization variables $\boldsymbol{\psi}$ and $\boldsymbol{\alpha}$ due to the definition of $h^{\mathsf{T}}_j$ in~\eqref{app_eq:target_hypothesis_def}. 
Using this fact, we can bound the hypothesis comparison term by expanding the definition of $h^{\mathsf{T}}_j$ as follows:
\begin{equation} \label{app_eq:T_ij_hcomparison_bound}
    \widehat{\varepsilon}^{\mathsf{T}}_j(h^{\mathsf{T}}_j,h^{\mathsf{S}}_i)
    \leq \sum_{k \in \mathcal{N}} (1-\psi_k)\psi_i\alpha_{k,i} \widehat{T}_{i,j,k},
\end{equation}
where $\widehat{T}_{i,j,k} = \frac{1}{\widehat{D}^{\mathsf{T}}_j} \sum_{x_d \in \widehat{D}^{\mathsf{T}}_j} \vert h^{\mathsf{S}}_{{k}}(x_d) - h^{\mathsf{S}}_i(x_d) \vert$. 
The result in~\eqref{app_eq:T_ij_hcomparison_bound} contains a negative variable $-\psi_k$, so we employ an auxiliary variable $\widehat{\chi}^{\mathsf{T}}_{i,j,k} > 0$ to bound the terms and convert them to constraints as follows:
\begin{align} \label{app_eq:aux_hypothesis_comp}
    & (1-\psi_k)\psi_i \alpha_{k,i} \widehat{T}_{i,j,k} \leq \widehat{\chi}^{\mathsf{T}}_{i,j,k} \\
    & \frac{1}{\psi_k + \frac{\widehat{\chi}^{\mathsf{T}}_{i,j,k}}{\psi_i \alpha_{k,i}\widehat{T}_{i,j,k}}} \leq 1, \label{app_eq:aux_hypothesis_comp_rearr}
\end{align}
where~\eqref{app_eq:aux_hypothesis_comp_rearr} is a result of rearranging~\eqref{app_eq:aux_hypothesis_comp}. 
Due to the posynomial in the denominator,~\eqref{app_eq:aux_hypothesis_comp_rearr} is not in the format of GP, and we exploit Lemma~\ref{Lemma:AG_mean} to approximate the denominator with a posynomial:
\begin{align} \label{app_eq:approx_term_b1}
    & G_{i,j,k}(\bm{x}) = \psi_k + \frac{\widehat{\chi}^{\mathsf{T}}_{i,j,k}}{\psi_i \alpha_{k,i}\widehat{T}_{i,j,k}} \geq \nonumber \\ 
    & 
    \widehat{G}_{i,j,k}(\bm{x};\ell) \triangleq \left( \frac{\psi_k G_{i,j,k}([\bm{x}]^{\ell-1})}{[\psi_k]^{\ell-1}} \right) ^ {\frac{[\psi_k]^{\ell-1}}{G_{i,j,k}([\bm{x}]^{\ell-1})}} 
    \left( \frac{ \frac{\widehat{\chi}^{\mathsf{T}}_{i,j,k}}{\psi_i \alpha_{k,i}} G_{i,j,k}([\bm{x}]^{\ell-1})} {\left[\frac{\widehat{\chi}^{\mathsf{T}}_{i,j,k}}{\psi_i \alpha_{k,i}}\right]^{\ell-1}} \right) ^ {\frac{\left[\frac{\widehat{\chi}^{\mathsf{T}}_{i,j,k}}{\widehat{T}_{i,j,k}\psi_i \alpha_{k,i}}\right]^{\ell-1}} {G_{i,j,k}([\bm{x}]^{\ell-1})}},
\end{align}
where we the fact that $\widehat{T}_{i,j,k}$ is a constant to obtain $\widehat{T}_{i,j,k} = [\widehat{T}_{i,j,k}]^{\ell-1}$. We finally approximate $(1-\psi_k)\psi_i \alpha_{k,i} \widehat{T}_{i,j,k}$ for each $i,j,k \in \mathcal{N}$ as follows: 
\begin{tcolorbox}[ams align]
& \frac{1}{\widehat{G}_{i,j,k}(\bm{x};\ell)} \leq 1, \forall i,j,k \in \mathcal{N}\\
& \widehat{\chi}^{\mathsf{T}}_{i,j,k} \geq 0, \forall i,j,k \in \mathcal{N}.
\end{tcolorbox}

We can now address the rest of inner summation terms in $(b)$ by grouping $T_{i,j}$ into constant and optimization variable terms as follows: 
\begin{equation} \label{app_eq:term_b_complete}
    \psi_j (1-\psi_i)\alpha_{i,j}T_{i,j} = \psi_j (1-\psi_i)\alpha_{i,j}
    \left( \widehat{T}_{i,j} 
    + \sum_{k \in \mathcal{N}} \widehat{\chi}^{\mathsf{T}}_{i,j,k} \right), 
\end{equation}
where 
{\color{black}
$\widehat{T}_{i,j} = \bigg( \widehat{\varepsilon}^{\mathsf{S}}_i(h^{\mathsf{S}}_i)+10\sqrt{2\log(2)}
+ \varepsilon^{\mathsf{T}}_j(f^{\mathsf{T}}_j,f^{\mathsf{S}}_i)
+ \frac{1}{2}\widehat{d}_{\mathcal{H}\Delta\mathcal{H}}(\widehat{\mathcal{D}}^{\mathsf{T}}_j,\widehat{\mathcal{D}}^{\mathsf{S}}_i)
+ 6 \left( \sqrt{ \log(2/\delta)/(2\widehat{D}^{\mathsf{S}}_i)}
+ \sqrt{ \log(2/\delta)/ (2\widehat{D}^{\mathsf{T}}_j)} \right) \bigg)$ is a constant and
$\sum_{k \in \mathcal{N}} \widehat{\chi}^{\mathsf{T}}_{i,j,k}$ is a variable. 
}
Due to the negative variable $-\psi_i$ in~\eqref{app_eq:term_b_complete}, which is not in the format of GP, we need to introduce another auxiliary variable $\chi^{\mathsf{T}}_{i,j} > 0$ as follows: 
\begin{align} \label{app_eq:complete_term_b}
    & \psi_j (1-\psi_i)\alpha_{i,j}
    \left( \widehat{T}_{i,j} 
    + \sum_{k \in \mathcal{N}} \widehat{\chi}^{\mathsf{T}}_{i,j,k} \right)
    \leq \chi^{\mathsf{T}}_{i,j} \\
    & \frac{\left( \widehat{T}_{i,j} + \sum_{k \in \mathcal{N}} \widehat{\chi}^{\mathsf{T}}_{i,j,k} \right) } {\psi_i\widehat{T}_{i,j} + \psi_i \sum_{k \in \mathcal{N}} \widehat{\chi}^{\mathsf{T}}_{i,j,k} + \frac{\chi^{\mathsf{T}}_{i,j}}{\psi_j\alpha_{i,j}}}
    \leq 1, \label{app_eq:complete_term_b_rearr}
\end{align} 
where~\eqref{app_eq:complete_term_b_rearr} is a result of rearranging~\eqref{app_eq:complete_term_b}. 
We now approximate the posynomial in the denominator of~\eqref{app_eq:complete_term_b}:
\begin{align} \label{app_eq:approx_term_b2}
    & H_{i,j}(\bm{x}) = \psi_i \widehat{T}_{i,j} + \psi_i \sum_{k \in \mathcal{N}} \widehat{\chi}^{\mathsf{T}}_{i,j,k} + \frac{\chi^{\mathsf{T}}_{i,j}}{\psi_j\alpha_{i,j}} \geq 
    \widehat{H}_{i,j}(\bm{x};\ell) \triangleq 
    \left( \frac{\psi_i H_{i,j}([\bm{x}]^{\ell-1})}{[\psi_i]^{\ell-1}}  \right) ^ { \frac{[\psi_i]^{\ell-1} \widehat{T}_{i,j}} {H_{i,j}([\bm{x}]^{\ell-1})}} 
    \nonumber \\
    &
    \left( \frac{ \frac{\chi^{\mathsf{T}}_{i,j}}{\psi_j\alpha_{i,j}} H_{i,j}([\bm{x}]^{\ell-1}) } {\left[\frac{\chi^{\mathsf{T}}_{i,j}}{\psi_j\alpha_{i,j}}\right]^{\ell-1}}   \right) ^ {\frac{\left[\frac{\chi^{\mathsf{T}}_{i,j}}{\psi_j\alpha_{i,j}}\right]^{\ell-1}} {H_{i,j}([\bm{x}]^{\ell-1})} } 
    \prod_{k \in \mathcal{N}} \left( \frac{\psi_i \widehat{\chi}^{\mathsf{T}}_{i,j,k} H_{i,j}([\bm{x}]^{\ell-1})} 
    {[\psi_i \widehat{\chi}^{\mathsf{T}}_{i,j,k}]^{\ell-1} } \right) ^{ \frac{[\psi_i \widehat{\chi}^{\mathsf{T}}_{i,j,k}]^{\ell-1}}{H_{i,j}([\bm{x}]^{\ell-1})}}.
\end{align}

Note that the optimization is significantly slower due to~\eqref{app_eq:approx_term_b2}, as the approximations for $\widehat{\varepsilon}^{\mathsf{T}}_{j}(h^{\mathsf{T}}_j,h^{\mathsf{S}}_i)$ introduces many more optimization variables, and we didn't observe any significant improvements from adding it. We therefore chose to omit $\widehat{\varepsilon}^{\mathsf{T}}_{j}(h^{\mathsf{T}}_j,h^{\mathsf{S}}_i)$ in the simulations. 
We finally approximate $\psi_j(1-\psi_i)\alpha_{i,j}T_{i,j}$ for each $i,j \in \mathcal{N}$ as follows:
\begin{tcolorbox}[ams align]
& \frac{\left( \widehat{T}_{i,j} + \sum_{k \in \mathcal{N}} \widehat{\chi}^{\mathsf{T}}_{i,j,k} \right)}{\widehat{H}_{i,j}(\bm{x};\ell)} \leq 1, \forall i,j \in \mathcal{N} \\
& \chi^{\mathsf{T}}_{i,j} \geq 0, \forall i,j \in \mathcal{N}.
\end{tcolorbox}

\textbf{Term $(c)$ in Objective Function.}
Term $(c)$ relies on the definition of $E_{i,j}$ in~\eqref{app_eq:def_E_i}, which contains a posynomial denominator and thus is not in the format of GP. We exploit Lemma~\ref{Lemma:AG_mean} to approximate the denominator as follows:
\begin{align} \label{app_eq:approx_term_c}
    J_{i,j}(\bm{x}) = \alpha_{i,j} + \epsilon_{E} 
    \geq \widehat{J}_{i,j}(\bm{x};\ell) 
    \triangleq \left( \frac{\alpha_{i,j} J_{i,j}([\bm{x}]^{\ell-1})}{[\alpha_{i,j}]^{\ell-1}} \right) ^ {\frac{[\alpha_{i,j}]^{\ell-1}}{J_{i,j}([\bm{x}]^{\ell-1})}} 
    \left( J_{i,j}([\bm{x}]^{\ell-1}) \right) ^ {\frac{\epsilon_E}{J_{i,j}([\bm{x}]^{\ell-1})}},
\end{align}
where we used $\epsilon_E = [\epsilon_E]^{\ell-1}$ as $\epsilon_E$ is a constant.
We can thus update $E_{i,j}$ as:
\begin{tcolorbox}[ams align]
& E_{i,j}(\alpha_{i,j}) = K_{i,j} \frac{\alpha_{i,j}}{\widehat{J}_{i,j}(\bm{x};\ell)}, \forall i,j \in \mathcal{N}.
\end{tcolorbox}

\textbf{Constraint~\eqref{app_eq:only_targets_receive}.}
Constraint~\eqref{app_eq:only_targets_receive} is an equality constraint on posynomials, which is not in the format of GP. To convert it into an inequality constraint, we introduce an auxiliary variable $\chi^{C}_{i,j} > 0$ for each $(i,j)$-pair, $i,j \in \mathcal{N}$ and a constant $\epsilon_{C} > 0$. We then introduce two replacement constraints as follows:
\begin{align}
    & \sum_{i \in \mathcal{N}} \alpha_{i,j} - \psi_j \leq \chi^{C}_{i,j} + \epsilon_{C} \label{app_eq:con_1_pos}\\
    & \sum_{i \in \mathcal{N}} \alpha_{i,j} - \psi_j \geq \chi^{C}_{i,j} - \epsilon_{C} \label{app_eq:con_1_neg}.
\end{align}
The expressions in ~\eqref{app_eq:con_1_pos} and~\eqref{app_eq:con_1_neg} contain negative variables, which again is not in the format of GP. 
Rearranging~\eqref{app_eq:con_1_pos} and~\eqref{app_eq:con_1_neg} yields:
\begin{align}
    & \frac{\sum_{i \in \mathcal{N}} \alpha_{i,j}}{ \chi^{C}_{i,j} + \epsilon_{C} + \psi_j} \leq 1 \label{app_eq:con_1_pos_rearr}\\
    & \frac{\chi^{C}_{i,j} - \epsilon_{C} + \psi_j}{\sum_{i \in \mathcal{N}} \alpha_{i,j}} \leq 1 \label{app_eq:con_1_neg_rearr}.
\end{align}
Both~\eqref{app_eq:con_1_pos_rearr} and~\eqref{app_eq:con_1_neg} contain posynomial denominators and are not in the format of GP. We thus exploit Lemma~\ref{Lemma:AG_mean} to obtain the following monomial approximations for the denominators:
\begin{align} \label{app_eq:posy_con_pos}
    & M^{+}_{i,j}(\bm{x}) = \chi^{C}_{i,j} + \epsilon_{C} + \psi_j \geq 
    \widehat{ M}^{+}_{i,j}(\bm{x};\ell) \triangleq \nonumber \\
    & 
    \left( \frac{\chi^{C}_{i,j} M^{+}_{i,j}([\bm{x}]^{\ell-1})}{[\chi^{C}_{i,j}]^{\ell-1}} \right) ^{\frac{[\chi^{C}_{i,j}]^{\ell-1}}{M^{+}_{i,j}([\bm{x}]^{\ell-1})}} 
    \left( M^{+}_{i,j}([\bm{x}]^{\ell-1}) \right) ^{\frac{\epsilon_{C}}{M^{+}_{i,j}([\bm{x}]^{\ell-1})}}
    \left( \frac{\psi_j M^{+}_{i,j}([\bm{x}]^{\ell-1})}{[\psi_j]^{\ell-1}} \right) ^{\frac{[\psi_j]^{\ell-1}}{M^{+}_{i,j}([\bm{x}]^{\ell-1})}},   
\end{align}
and 
\begin{align} \label{app_eq:posy_con_neg}
    & M^{-}_{i,j}(\bm{x}) = \sum_{i \in \mathcal{N}} \alpha_{i,j} \geq 
    \widehat{ M}^{-}_{i,j}(\bm{x};\ell) \triangleq 
    \prod_{i \in \mathcal{N}} \left( \frac{\alpha_{i,j} M^{-}_{i,j}([\bm{x}]^{\ell-1})}{[\alpha_{i,j}]^{\ell-1}} \right) ^{\frac{[\alpha_{i,j}]^{\ell-1}}{M^{-}_{i,j}([\bm{x}]^{\ell-1})}}, 
\end{align}
where we use~\eqref{app_eq:posy_con_pos} and~\eqref{app_eq:posy_con_neg} to approximate the denominators in~\eqref{app_eq:con_1_pos_rearr} and~\eqref{app_eq:con_1_neg_rearr} respectively. Using this result we now have the constraints:
\begin{tcolorbox}[ams align]
& \frac{\sum_{i \in \mathcal{N}} \alpha_{i,j}}{ \widehat{M}^{+}_{i,j}(\bm{x};\ell) } \leq 1 \\
& \frac{\chi^{C}_{i,j} - \epsilon_{C} + \psi_j}{\widehat{M}^{-}_{i,j}(\bm{x};\ell)} \leq 1.
\end{tcolorbox}

We note that for the simulations in Sec.~\ref{sec:experiments}, we chose $\epsilon_C = 1e-2$.

Combining all the aforementioned adjustments to $(\boldsymbol{\mathcal{P}})$ yields our final optimization formulation $(\boldsymbol{\mathcal{P}}^{'})$, which is an iterative optimization, where at each iteration, the optimization admits the GP form. 
\begin{tcolorbox}[ams align]
    & (\boldsymbol{\mathcal{P}}^{'}):~\argmin
    \phi^{\mathsf{S}} \sum_{i \in \mathcal{N}} \chi^{\mathsf{S}}_i 
    + \phi^{\mathsf{T}} \sum_{j \in \mathcal{N}} \sum_{i \in \mathcal{N}} \chi^{\mathsf{T}}_{i,j}
    + \phi^{\mathsf{E}} \sum_{i \in \mathcal{N}} \sum_{j \in \mathcal{N}} E_{i,j} (\alpha_{i,j}) 
    + \sum_{j \in \mathcal{N}} \sum_{i \in \mathcal{N}} \chi^{C}_{i,j}
    \label{app_eq:prob_prime}
    \\
    & \textrm{s.t.} \nonumber \\
    & h_j^{\mathsf{T}} = \sum_{i \in \mathcal{N}} \alpha_{i,j} (1-\psi_i) \psi_j h_i^{\mathsf{S}}, \forall j \in \mathcal{N} \\
    & E_{i,j}(\alpha_{i,j}) = K_{i,j} \frac{\alpha_{i,j}}{ \widehat{J}_{i,j}(\bm{x};\ell) }, \forall i,j \in \mathcal{N} \\ 
    & \frac{1}{\widehat{F}_i(\bm{x};\ell)} \leq 1, \forall i \in \mathcal{N} \\
    & \frac{1}{\widehat{G}_{i,j,k}(\bm{x};\ell)} \leq 1, \forall i,j,k \in \mathcal{N} \\
    & \frac{\left( \widehat{T}_{i,j} + \sum_{k \in \mathcal{N}} \widehat{\chi}^{\mathsf{T}}_{i,j,k} \right)}{\widehat{H}_{i,j}(\bm{x};\ell)} \leq 1, \forall i,j \in \mathcal{N} \\ 
    & \frac{\sum_{i \in \mathcal{N}} \alpha_{i,j}}{ \widehat{M}^{+}_{i,j}(\bm{x};\ell) } \leq 1 \\
    & \frac{\chi^{C}_{i,j} - \epsilon_{C} + \psi_j}{\widehat{M}^{-}_{i,j}(\bm{x};\ell)} \leq 1 \\
    & 0 \leq \alpha_{i,j} \leq 1, \psi_i(t) \in \{ 0,1 \}, \forall i,j \in \mathcal{N} \\
    & \chi^{\mathsf{S}}_i > 0, \widehat{\chi}^{\mathsf{T}}_{i,j,k} \geq 0, \chi^{\mathsf{T}}_{i,j} > 0, \chi^{C}_{i,j} > 0, \forall i,j,k \in \mathcal{N} 
    \label{app_eq:prob_prime_lasteq}
    \\
    \cline{1-2}
    & \textrm{\textit{\textbf{Variables}}:}~\boldsymbol{\alpha},\boldsymbol{\psi},
    \{ \chi^{\mathsf{S}}_i, \widehat{\chi}^{\mathsf{T}}_{i,j,k}, \chi^{\mathsf{T}}_{i,j}, \chi^{C}_{i,j} \}_{i,j,k \in \mathcal{N}} 
\end{tcolorbox}



\newpage
\end{document}